\documentclass[lettersize,journal]{IEEEtran}
\usepackage{algorithmic}
\usepackage{algorithm}
\usepackage{array}
\usepackage{booktabs} 
\usepackage{soul}
\usepackage{amsmath,amsfonts,amssymb}
\usepackage[caption=false,font=footnotesize]{subfig}
\usepackage{textcomp}
\usepackage{stfloats}
\usepackage{bbm}
\usepackage{url}
\usepackage{verbatim}
\usepackage{graphicx}
\usepackage[style=ieee,backend=biber,citestyle=ieee-comp]{biblatex}
\DeclareFieldFormat{doi}{} 
\usepackage{cleveref}
\crefname{figure}{Fig.}{Figs.}
\Crefname{figure}{Fig.}{Figs.}
\crefname{equation}{Eq.}{Eqs.}
\Crefname{equation}{Eq.}{Eqs.}
\crefname{appendix}{Appendix}{Appendices}
\Crefname{appendix}{Appendix}{Appendices}
\usepackage{amsthm}
\usepackage{pgfplots}
\pgfplotsset{compat=1.18}

\newtheoremstyle{mylemma}
  {0pt}      
  {0pt}      
  {\itshape} 
  {}         
  {\bfseries}
  {:}        
  {0.5em}    
  {\thmname{#1}\thmnumber{ #2}}

\theoremstyle{mylemma}
\newtheorem{lemma}{Lemma}
\crefname{lemma}{Lemma}{Lemmas}
\Crefname{lemma}{Lemma}{Lemmas}
\usepackage{xcolor}

\makeatletter

\def\section{\@startsection{section}{1}{\z@}%
{1.5ex plus 0.2ex minus 0.2ex}%
{0.5ex plus 0.1ex minus 0.1ex}%
{\normalfont\normalsize\centering\scshape}}

\def\subsection{\@startsection{subsection}{2}{\z@}%
{1.2ex plus 0.2ex minus 0.2ex}%
{0.4ex plus 0.1ex minus 0.1ex}%
{\normalfont\normalsize\itshape}}

\def\subsubsection{\@startsection{subsubsection}{3}{\parindent}%
{0ex plus 0.1ex minus 0.1ex}%
{0ex}%
{\normalfont\normalsize\itshape}}
\def\@IEEEfiguretopskipspace{}
\makeatother

\begin{document}

\title{Adaptive Payload-Aided Near-Field Beam Tracking via Thompson Sampling}

\author{Junchi Liu, Zijun Wang, Shawn Tsai, Rui Zhang,~\IEEEmembership{Member,~IEEE}
\thanks{Junchi Liu, Zijun Wang and Rui Zhang are with the Department of Electrical Engineering, The State University of New York at Buffalo, New York, USA (email: {junchili@buffalo.edu}, {zwang267@buffalo.edu}, {rzhang45@buffalo.edu}).}
\thanks{Shawn Tsai is with the CSD, MediaTek Inc. USA, San Diego, CA 92122 USA (e-mail: shawn.tsai@mediatek.com). }
\thanks{This work was supported in part by the CSD, MediaTek Inc. USA, under Grant 103764, and in part by the National Science Foundation, under Grant ECCS 2512911, IIS-2549124.}
\thanks{Corresponding author: Rui Zhang.}}

\markboth{IEEE Transactions on Wireless Communications}%
{Liu \MakeLowercase{\textit{et al.}}: Adaptive Payload-Aided Near-Field Beam Tracking via Thompson Sampling}


\maketitle

\begin{abstract}
Extremely large antenna arrays operating at high frequencies substantially extend the radiative near-field region in 6G networks, making mobile beam alignment depend jointly on the user's
angle and range. 
The additional range dimension
enlarges the beam-search space and makes repeated pilot-based beam
sweeping particularly costly under mobility, while sensing-assisted
tracking relies on echo observations whose quality depends on the
propagation environment and target reflectivity. To address these
limitations, we propose an adaptive payload-aided near-field
beam-tracking framework based on maximum likelihood estimation (MLE)
and Thompson sampling (TS). Selected received payload samples are fed back by the user and reused as tracking observations, thereby avoiding dedicated non-payload beam-sweeping symbols during the tracking stage. Within each sliding observation window, the local angle and range trajectories are represented by low-order polynomials, thereby capturing short-term velocity, acceleration, and higher-order motion variations, and their
coefficients are estimated through MLE using the spherical-wave channel model.
A local Gaussian approximation centered at the maximum-likelihood estimate, with covariance given by the inverse observed Fisher information, provides a tractable local representation of trajectory uncertainty. 
For payload transmissions selected for receiver feedback, TS draws a trajectory hypothesis and maps it to a payload beam directed toward the sampled state; the remaining transmissions use the MLE-predicted beam for exploitation. 
To handle nonstationary mobility, an asymptotic chi-square characterization of in-window estimation risk motivates the coordination of the observation-window length and polynomial degree, while online residual statistics guide the adaptation of the update interval and feedback ratio.
The framework is further generalized to uniform planar arrays by incorporating elevation tracking. Simulations under smooth and sharp-turn trajectories demonstrate high payload-accounted mean normalized beamforming gain, low payload-accounted normalized-gain variance, and high effective-symbol reliability, while the adaptive mechanism provides substantial robustness gains under nonstationary sharp-turn mobility.
\end{abstract}


\section{Introduction}
\label{sec:intro}
\IEEEPARstart{T}{he} deployment of extremely large aperture arrays (ELAAs) and the
shift toward higher frequency bands, such as millimeter-wave and
terahertz bands, are key enablers for future wireless networks.
As the physical aperture of an array increases, the radiative
near-field region expands and may cover a substantial portion of the
typical service area of user equipment (UE)
\cite{Tutorial_with_beam_pattern_and_resolution,General_tutorial}. Unlike the far-field steering vector, which is primarily determined by the angular direction, the near-field steering vector depends jointly on the angle and range of the UE \cite{far-field-assumption}. Consequently,
near-field beam acquisition must resolve an additional radial
dimension.  Although structured codebooks can reduce the
number of tested beams, the additional range dimension fundamentally
enlarges the beam-management space. This enlargement becomes particularly consequential under mobility. Both the angular and radial coordinates evolve over time, while the narrow near-field beams generated by ELAAs make the beamforming gain highly sensitive to state mismatch. Moreover, under spherical-wave propagation, the propagation distance and its temporal derivative can differ across antenna elements, leading to spatially non-uniform mobility-induced channel-phase evolution and Doppler effects \cite{doppler_on_h}. Periodically repeating an angle--range beam sweep can therefore consume substantial non-payload resources and interrupt
data transmission. Near-field beam tracking is needed to exploit
temporal correlation and maintain alignment without repeatedly solving the enlarged beam-acquisition problem from scratch.

A variety of polar-domain codebooks have been developed for near-field channel estimation and beam training to reduce acquisition overhead \cite{kaustnearcb,cui2022farornf,chen2025beamspace}. Staged schemes first identify angular candidates and subsequently refine the range
\cite{zhang2022fastNFBT}, while DFT-based methods exploit near-field
beam patterns to infer angle and range jointly \cite{Joint_Angle_and_Range_Estimation}. Other approaches construct
spatial-chirp or slope-intercept ($k$--$b$)-domain codebooks \cite{k-b_domain_codebook}, hierarchical codebooks covering different
portions of the Fresnel region \cite{ChenHierarchical2023,weng2024nearfield,LuHierarchical2024,WuTwoStage2024}, or sparse DFT codebooks that resolve angular ambiguity before radial refinement \cite{ZhouSparseDFT2025}. These methods substantially reduce the cost of beam acquisition at a given UE location. Under mobility, however, beam refinement must be invoked repeatedly as both angle and range change.

 Beam tracking has been studied more extensively for far-field channels. Existing methods include particle-filter tracking with adaptive beamwidth control \cite{chung2021adaptive}; dual-timescale Kalman tracking of slowly varying angles of departure and arrival with abrupt-change detection \cite{Zhang2016AngleTracking}; sparse Bayesian channel tracking with Kalman filtering and smoothing, coupled with linear TS for selecting training beams \cite{Love_paper}; auxiliary-beam-pair angle tracking \cite{Zhu2018HighResolutionTracking}; and learned
continuous-time beam predictors based on historical measurements
\cite{LNN_tracking}. These methods exploit temporal observations to
avoid repeated full angular searches, but they generally operate on an angle-only far-field state and do not explicitly address the coupled angle--range dynamics of a mobile near-field channel. 

 Several recent works have considered near-field beam tracking directly.
The method in \cite{UPA_NF} estimates the effective beam coherence time from the predicted beamforming gain and triggers near-field beam sweeping when the gain is expected to fall below a prescribed threshold. A dynamic non-uniform coordinate grid is then constructed around the predicted user region to support local beam refinement. In \cite{sub_array_li}, a subarray architecture is combined with an extended Kalman filter (EKF) to predict the user's position and velocity. One uplink pilot is collected in each tracking cycle to update the state estimate and the subsequent beam. The method in \cite{YiMobility2025} uses prior information, including the deterministic trajectory of a high-speed railway scenario, to predict the beam direction and channel amplitude, and employs an error-detection module to correct prediction errors.
Online learning-based methods have also been introduced for near-field beam tracking. In \cite{wang2025dl}, an LSTM uses historical beam measurements to predict the angle and range in the next period, followed by a local beam search to update the tracking result. In \cite{park2024thz}, near-field beam tracking is formulated as a sequential decision problem, and a deep Q-network learns beam-update actions from online interaction.  These methods reduce the search region or improve robustness to mobility, but their signaling and computational requirements remain coupled to dedicated pilots, local beam searches, or online learning interactions. 

 Sensing-assisted predictive beamforming provides another means of reducing communication-pilot overhead. The method in \cite{tvt2024tracking} estimates radial and transverse velocities from echo observations and predicts the UE position for beamforming in the subsequent coherent processing interval. Near-field sensing has further been extended to full-motion-state estimation and multi-interval tracking using alternating optimization, nonlinear filtering, and learning methods
\cite{jiang2024nise,jiang2025niseframework}. Such schemes can reduce
communication beam sweeping, but their reliability depends on the
quality of sensing echoes, which can be affected by propagation loss, clutter, sensing geometry, and the radar cross section (RCS) or reflectivity of the sensing target.

The above approaches leave an important operating point insufficiently addressed: maintaining near-field alignment from communication observations without repeatedly inserting dedicated beam-sweeping intervals or relying on sensing echoes. In such a payload-aided scheme, each transmit beam affects both the instantaneous payload beamforming gain and the informativeness of the resulting feedback observation.
Steering exclusively toward the current point estimate is effective when the estimate is accurate. However, such pure exploitation may provide limited information about alternative angle--range locations. Consequently, prediction errors caused by abrupt motion or local-model mismatch may accumulate rather than being corrected and eventually cause track loss. Conversely, broadly probing the angle--range domain improves tracking information at the cost of instantaneous beamforming gain. This creates a communication--information tradeoff between exploiting the current trajectory estimate and exploring its uncertainty. Although linear TS was previously used in \cite{Love_paper} to select training beams from a discrete feasible beam set while a sparse angular channel was tracked through sparse Bayesian learning and Kalman filtering and smoothing, our setting differs in two respects: the unknown state is a continuous near-field trajectory in coupled angle and range, and every probing action is a payload-carrying beam. The tracker must therefore map posterior uncertainty in continuous trajectory parameters to payload-beam decisions while preserving communication connectivity.

To address this problem, we propose an adaptive payload-aided
near-field beam-tracking framework based on maximum likelihood
estimation (MLE) and TS. Selected
received payload samples are fed back by the UE and reused for local
trajectory estimation, uncertainty-aware beam probing, and online
tracking-parameter adaptation. Thus, no dedicated non-payload
beam-sweeping symbols are inserted during steady-state tracking.
Our preliminary study \cite{asilomar} introduced the
 constant-parameter ULA version of the MLE-based TS tracker, in
 which the observation-window length, polynomial degree, update
interval, and feedback ratio were fixed. 
Compared with
\cite{asilomar}, this article develops an adaptive tracking-parameter
mechanism, provides additional analytical characterizations of the
Gaussian posterior approximation and  an
asymptotic characterization of in-window estimation-risk
scaling, and extends the framework from two-dimensional ULA
tracking to three-dimensional UPA tracking.

The main contributions of this paper are summarized as follows.

\begin{itemize}
\setlength{\itemsep}{0pt}
\setlength{\parsep}{0pt}
\setlength{\parskip}{0pt}
    \item   We develop a payload-aided local trajectory estimator for near-field beam tracking. Within each sliding observation window, the time-varying UE angle and range are represented by low-order polynomials, whose coefficients are estimated from selected received payload samples through MLE. The temporal derivatives of the polynomial trajectories can represent local velocity, acceleration, and higher-order motion variations. Embedding these trajectories into the spherical-wave channel model also captures the associated element-dependent channel-phase evolution without introducing a separate Doppler state.

    \item  We develop an uncertainty-aware payload-beam selection strategy based on TS. A Gaussian approximation to the trajectory parameter posterior is constructed around the MLE, with covariance determined by the inverse observed Fisher information matrix. For payload symbols selected for receiver feedback, posterior samples are mapped to probing beams that acquire information about uncertain directions; the remaining payload symbols use the MLE-predicted beam for exploitation. This design integrates information acquisition with payload delivery without inserting separate downlink beam-sweeping symbols.

    \item   We develop an adaptive tracking-parameter mechanism for nonstationary mobility. An asymptotic chi-square relation provides an in-window estimation-risk budget for selecting model complexity.
    This relation motivates the coordination of the
    observation-window length and polynomial model dimension, while empirical
    residual statistics serve as online indicators of model mismatch
    and temporal prediction deterioration when adapting the update interval and feedback ratio.
    
    \item We generalize the proposed framework from ULA-based two-dimensional angle--range tracking to UPA-based three-dimensional azimuth--elevation--range tracking. Simulations under smooth and sharp-turn trajectories demonstrate high payload-accounted mean normalized beamforming gain, low payload-accounted normalized-gain variance, and high effective-symbol reliability, with substantial robustness gains under nonstationary mobility with sharp-turn trajectories.
\end{itemize}

The rest of this paper is organized as follows. \Cref{sec:sysprob} presents the near-field channel model and formulates the payload-aided
beam tracking problem. \Cref{sec:method} develops the MLE-based trajectory estimator and the constant-parameter TS tracking scheme. \Cref{sec:adaptive} presents the adaptive
tracking-parameter adjustment mechanism.
\Cref{sec:sim} evaluates the
proposed methods under the ULA setting and compares them with
representative baselines. \Cref{sec:upa} generalizes the framework to
the UPA architecture and presents the corresponding numerical results. Finally, \Cref{sec:con} concludes this paper.

\section{System Model and Problem Formulation}
\label{sec:sysprob}
This section introduces the near-field channel model for a mobile UE and formulates the beam tracking problem.

\subsection{Channel Model}\label{subsec:channel}
We consider a near-field multiple-input single-output (MISO) downlink system as illustrated in Fig.~\ref{channel}, where the base station (BS) is equipped with an $N$-element uniform linear array (ULA) and the moving UE has a single antenna.
Following the criterion in \cite{fresnel}, we characterize the near-field region as the distance interval between the Fresnel distance and the Rayleigh distance, i.e.,
$R_{\mathrm{Fre}}=\frac{1}{2}\sqrt{\frac{D^{3}}{\lambda}}$ and $R_{\mathrm{Ray}}=\frac{2D^{2}}{\lambda}$,
respectively. Here, $\lambda$ denotes the carrier wavelength, and $D=(N-1)d$ represents the physical aperture of the ULA with inter-element spacing $d$.
At time $t$, the BS transmits a payload or pilot symbol $x_t$ using a beamformer $\mathbf{w}_t\in\mathbb{C}^{N\times 1}$, and the UE observes
\( y_t = \mathbf{h}_{t,ULA}^{\mathsf H}\mathbf{w}_t x_t + n_t, \)
where $n_t\sim\mathcal{CN}(0,\sigma^2)$ and $\mathbf{h}_{t,ULA}\in\mathbb{C}^{N}$ is the channel vector at $t$. 

\begin{figure}[t]
\centering
\includegraphics[width=\linewidth]{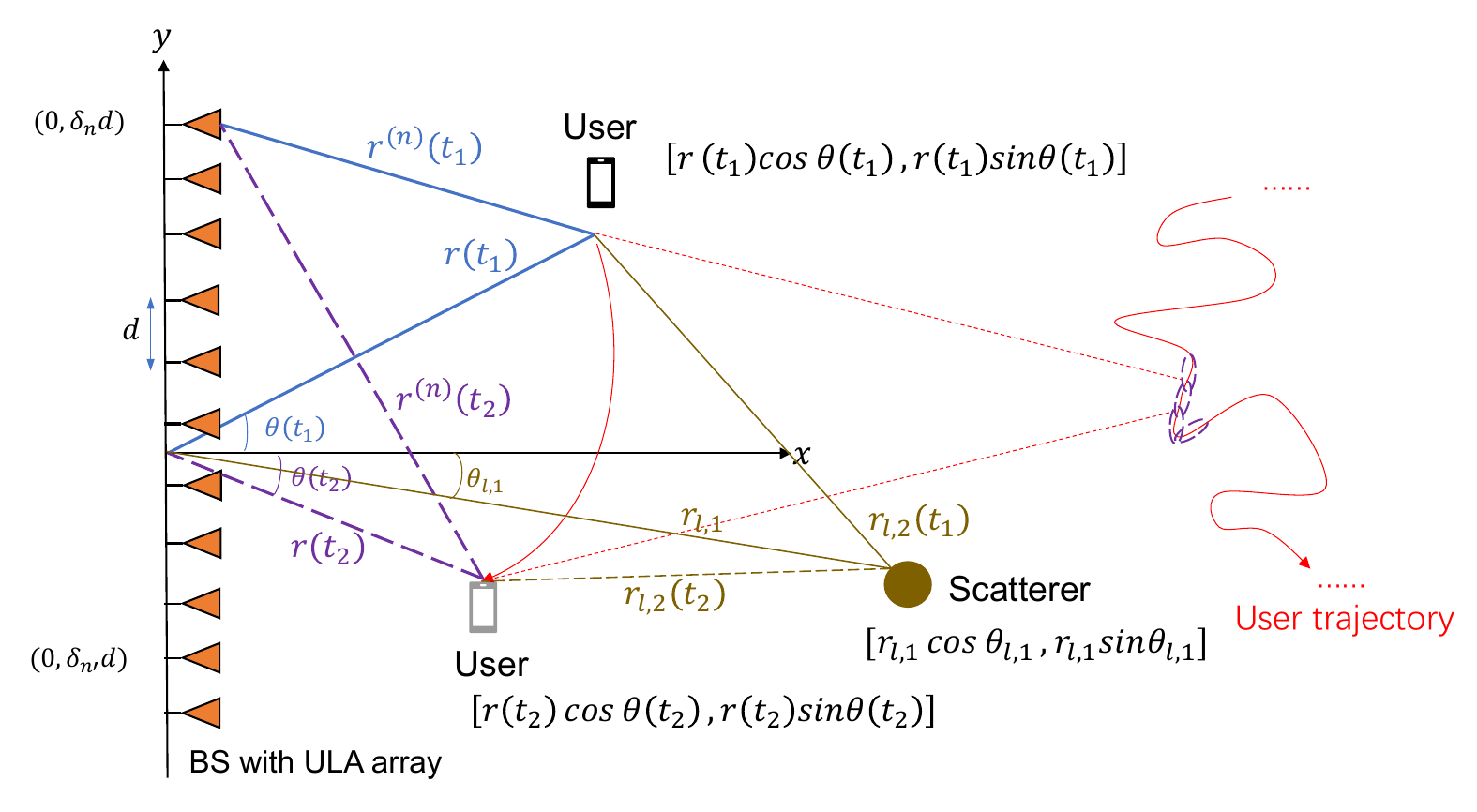}
\caption{ULA-based near-field channel with a mobile UE and a stationary scatterer.}
\label{channel}
\end{figure}

As depicted in Fig.~\ref{channel}, the BS employs a ULA aligned with the $y$-axis and centered at the origin $(0,0)$. The position of the $n$-th antenna element is given by $(0,\delta_n d)$, where 
$\delta_n \triangleq \frac{2n-N+1}{2}$, $n\in\mathcal{N}$, and $\mathcal{N}\triangleq\{0,1,\ldots,N-1\}$. 
The inter-element spacing is set to $d=\lambda/2$. 
The line-of-sight (LoS) channel vector $\mathbf{h}^{\mathrm{LoS}}_{t,\mathrm{ULA}}$ is determined by the UE's time-varying polar coordinates $[\theta(t),r(t)]$, where $r(t)$ and $\theta(t)$ denote its instantaneous range and angle relative to the array center. We define the near-field steering vector

\begin{IEEEeqnarray}{rCl}
\mathbf{b}_{ULA}\!\big(\theta(t),r(t)\big)
&=& \frac{1}{\sqrt{N}}
\Big[
e^{-j\frac{2\pi}{\lambda}\big(r^{(0)}(t)-r(t)\big)},\nonumber\\
&&\; \ldots,\ 
e^{-j\frac{2\pi}{\lambda}\big(r^{(N-1)}(t)-r(t)\big)}
\Big]^{\mathsf T}.
\end{IEEEeqnarray}
where the propagation distance from the $n$th element to the UE is
$r^{(n)}(t)=\sqrt{r^2(t)+\delta_n^2d^2-2r(t)\sin(\theta(t))\delta_nd}$.

Accordingly, the BS--UE LoS channel vector is
\( \mathbf{h}^{\mathrm{LoS}}_{t,\mathrm{ULA}} =g(t)e^{-j\frac{2\pi r(t)}{\lambda}}\mathbf b_{ULA}(\theta(t),r(t)), \)
where $g(t)\triangleq\lambda/(4\pi r(t))$ captures the free-space path gain \cite{sun2016pathloss}. Equivalently, its $n$th entry is
$h^{\mathrm{LoS}}_{t,\mathrm{ULA}}(n)=g(t)e^{-j2\pi r^{(n)}(t)/\lambda}/\sqrt N$.
Thus, the element-dependent Doppler shift follows directly from the distance derivative, \( f_{D,n}(t)=-\lambda^{-1}\mathrm d r^{(n)}(t)/\mathrm dt, \) and is not introduced as a separate state.
We assume that the scatterers remain stationary, while the scatterer--UE segment varies with the moving UE. The NLoS component is
\( \mathbf h^{\mathrm{NLoS}}_{t,\mathrm{ULA}} =\sum_{l=1}^{L-1}g_l(t)e^{-j\frac{2\pi[r_{l,1}+r_{l,2}(t)]}{\lambda}} \mathbf b_{ULA}(\theta_{l,1},r_{l,1}), \)
where the $l$-th scatterer is at $(\theta_{l,1},r_{l,1})$ relative to the array center,
$r_{l,2}(t)=\sqrt{r^2(t)+r_{l,1}^2-2r(t)r_{l,1}\cos[\theta(t)-\theta_{l,1}]}$, and
$g_l(t)=\lambda p_l/[4\pi (r_{l,1} + r_{l,2}(t))]$ includes path loss and reflection coefficient $p_l$. The complete channel is $\mathbf h_{t,\mathrm{ULA}}=\mathbf h^{\mathrm{LoS}}_{t,\mathrm{ULA}}+\mathbf h^{\mathrm{NLoS}}_{t,\mathrm{ULA}}$.
With unit-energy transmit symbols, the instantaneous pre-beamforming SNR at time $t$ is defined as $\mathrm{SNR}_t\triangleq\|\mathbf h_{t,\mathrm{ULA}}\|_2^2/(N\sigma^2)$, or equivalently, $\mathrm{SNR}_{t,\mathrm{dB}}\triangleq10\log_{10}\!\left(\|\mathbf h_{t,\mathrm{ULA}}\|_2^2/(N\sigma^2)\right)$. Here, $\sigma^2$ is the time-invariant noise variance, whereas $\mathrm{SNR}_t$ generally varies with $t$ as the channel power changes.
\subsection{Problem Formulation}\label{subsec:prob-form}

Because the LoS component is typically dominant in the considered
high-frequency setting \cite{NYU_NLOS}, the estimator models the time-varying LoS path, 
i.e., $\widehat{\mathbf h}_{t}=\widehat{\mathbf h}^{\mathrm{LoS}}_{t}$.
The objective is to track the time-varying polar coordinates $(\theta(t),r(t))$ and update the unit-norm transmit beamformer $\mathbf w_t$ accordingly. For a channel $\mathbf h_t$, define the normalized beamforming gain
\begin{equation}\label{eq:normalized_gain}
G_t(\mathbf w_t)\triangleq |\mathbf h_t^{\mathsf H}\mathbf w_t|^2 / \|\mathbf h_t\|_2^2,
\quad \|\mathbf w_t\|_2=1.
\end{equation}
The tracking policy seeks to maximize the expected time-average of \Cref{eq:normalized_gain}, subject to pilot overhead and receiver-feedback budget. 

\begin{figure}[!t]
\centering
\includegraphics[width=\linewidth]{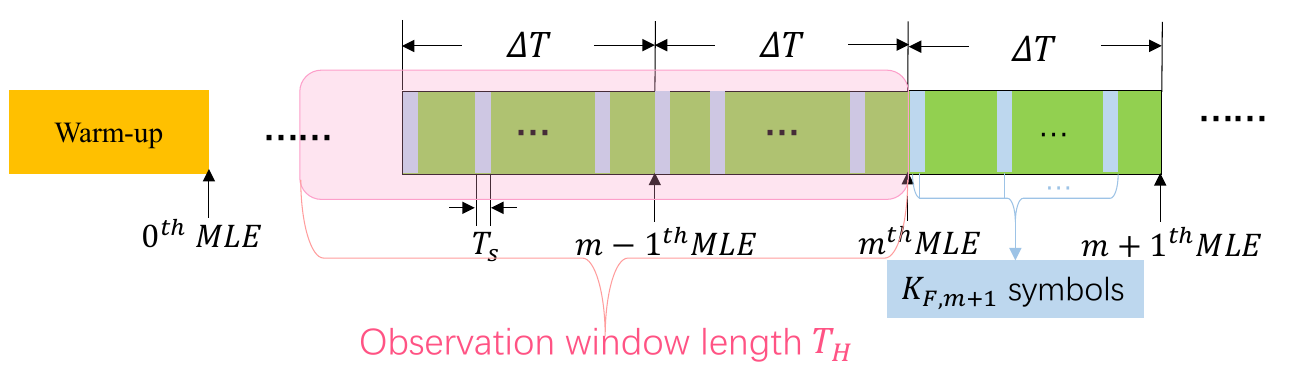} 
\caption{Proposed beam tracking protocol.}
\label{fig:tracking_protocol}
\end{figure}


At the end of the $m$-th update interval, let $t_m$ denote the physical update time and let $T_{H,m}$ denote the observation-window length. For a feedback observation acquired at physical time $t_u\in[t_m-T_{H,m},t_m]$, we use the window-relative, unnormalized time
\( t'_{u,m}\triangleq t_u-(t_m-T_{H,m})\in[0,T_{H,m}]. \) The same coordinate is used for
extrapolation during the $(m+1)$-th interval, for which
$t'_{u,m}>T_{H,m}$.
The coefficients are re-estimated in this local coordinate whenever the window changes. Assuming local trajectory smoothness, the angle and range at a generic local time $t'$ are represented by
\begin{equation}\label{eq:poly_apprx}
\widehat\theta_m(t')=\sum_{i=0}^{p_{\theta,m}}\alpha_{m,i}t'^i,
\qquad
\widehat r_m(t')=\sum_{i=0}^{p_{r,m}}\beta_{m,i}t'^i,
\end{equation}
where $p_{\theta,m}$ and $p_{r,m}$ are polynomial degrees,
$\boldsymbol\alpha_m=[\alpha_{m,0},\ldots,\alpha_{m,p_{\theta,m}}]^{\mathsf T}$,
$\boldsymbol\beta_m=[\beta_{m,0},\ldots,\beta_{m,p_{r,m}}]^{\mathsf T}$, and
\( \boldsymbol\eta_m\triangleq [\boldsymbol\alpha_m^{\mathsf T},\boldsymbol\beta_m^{\mathsf T}]^{\mathsf T} \in\mathbb R^{p_m}, \quad p_m=p_{\theta,m}+p_{r,m}+2. \)
Substituting \Cref{eq:poly_apprx} into the LoS model gives
\begin{equation}\label{eq:h_est}
\widehat{\mathbf h}^{\mathrm{LoS}}_{m,\mathrm{ULA}}(t';\boldsymbol\eta_m)
=\widehat g_m(t')e^{-j\frac{2\pi\widehat r_m(t')}{\lambda}}
\mathbf b_{ULA}\!\left(\widehat\theta_m(t'),\widehat r_m(t')\right).
\end{equation}
where $\widehat g_m(t')\triangleq\lambda/[4\pi\widehat r_m(t')]$. Below, the LoS superscript is omitted from estimated channels for brevity.

During warm-up, a fixed beamformer constructed from the UE's known initial position is transmitted for 256 consecutive symbols while the UE moves, and the resulting observations initialize the first MLE. Following the warm-up phase, the system enters the tracking protocol in \Cref{fig:tracking_protocol}. Let $t_0=0$ and
\( t_m\triangleq\sum_{j=1}^{m}\Delta T_j, \)
so that the \(m\)-th interval spans $[t_{m-1},t_m)$. With symbol duration $T_s$, we define the number of payload symbols transmitted during the \(m\)-th interval as \(K_m\triangleq\max\!\left\{2,\left\lfloor\frac{\Delta T_m}{T_s}\right\rfloor\right\}\), the cumulative number of symbols transmitted before the \(m\)-th interval as \(U_m^{\mathrm{start}}\triangleq\sum_{j=1}^{m-1}K_j\), and the corresponding global symbol index as \(u(m,k)\triangleq U_m^{\mathrm{start}}+k\) for $k\in\{1,\ldots,K_m\}$. Its physical time is $t_{m,k}\triangleq t_{m-1}+(k-1)T_s$, and we write $t_{u(m,k)}=t_{m,k}$.




The feedback ratio $\omega_m\in[0,1]$ determines the number
$K_{F,m}$ of feedback-designated payload symbols in interval $m$:
\( K_{F,m}\triangleq\min\!\left\{K_m, \max\!\left\{2,\left\lceil\omega_mK_m\right\rceil\right\}\right\}. \)
To cover the interval approximately uniformly, its feedback-position set is
\begin{equation}\label{eq:feedback_set}
\mathcal K_{F,m}\triangleq\left\{
1+\left\lfloor\frac{(i-1)(K_m-1)}{K_{F,m}-1}\right\rfloor
:i=1,\ldots,K_{F,m}\right\}.
\end{equation}
When $k\in\mathcal K_{F,m}$, the UE returns the received payload sample indexed by $u(m,k)$.

At the end of interval $m$, we define the symbol-index set used for MLE as $\mathcal U_m\triangleq \{u:t_m-T_{H,m}\le t_u\le t_m,\ u\text{ indexes a feedback sample}\}$ and the corresponding history as $ \mathcal H_m\triangleq \{(t_u,x_u,\mathbf w_u,y_u)\}_{u\in\mathcal U_m}. $
For $u\in\mathcal U_m$, evaluating \Cref{eq:h_est} at $t'=t'_{u,m}$ gives the noiseless predicted observation
\( \mu_u(\boldsymbol\eta) \triangleq \widehat{\mathbf h}_{m,\mathrm{ULA}}^{\mathsf H}(t'_{u,m};\boldsymbol\eta) \mathbf w_ux_u. \)
Under AWGN with variance $\sigma^2$, the conditional density is  $p\!\left(y_{u}\mid \boldsymbol{\eta}\right) =(\pi\sigma^2)^{-1}\exp\!\left( -\left|y_{u}-\mu_{u}(\boldsymbol{\eta})\right|^2/\sigma^2 \right). $

The joint likelihood is

\begin{equation}\label{eq:likelihood}
\begin{split}
\mathcal{L}_m(\boldsymbol{\eta})
&\triangleq \prod_{u\in\mathcal{U}_m} p\!\left(y_{u}\mid \boldsymbol{\eta}\right) \\
&= (\pi\sigma^2)^{-|\mathcal{U}_m|}  \exp\!\Bigg( -\frac{1}{\sigma^2}\sum_{u\in\mathcal{U}_m} 
\left|y_{u}-\mu_{u}(\boldsymbol{\eta})\right|^2 \Bigg).
\end{split}
\end{equation}

Maximizing the likelihood function in Eq. \eqref{eq:likelihood} is equivalent to minimizing the negative log-likelihood, resulting in the following nonlinear least-squares problem:
\begin{equation}\label{eq:mle_problem}
\hat{\boldsymbol{\eta}}_m
= \arg\min_{\boldsymbol{\eta}} \;
J^{ULA}_m(\boldsymbol{\eta}),
\end{equation}
where $J^{ULA}_m(\boldsymbol{\eta})
\triangleq \sum_{u\in\mathcal{U}_m} \left|y_{u}-\mu_{u}(\boldsymbol{\eta})\right|^2 $.


After collecting the feedback from the $m$-th interval, the BS estimates the local trajectory parameters by solving \Cref{eq:mle_problem}.
Because the channel depends nonlinearly on the polynomial coefficients, $J_m^{\mathrm{ULA}}(\boldsymbol\eta)$ is generally nonconvex. We solve \Cref{eq:mle_problem} with Adam \cite{kingma2015adam}, initialized by the preceding interval's estimate. 
Since the local-time origin changes
between updates, the preceding coefficients are first translated
to the current coordinate system.
Let
\(
\delta_m
\triangleq
\left(t_m-T_{H,m}\right)
-
\left(t_{m-1}-T_{H,m-1}\right).
\)
The translated initialization is
\(
\alpha_{m,\ell}^{\mathrm{init}}
=
\sum_{i=\ell}^{p_{\theta,m-1}}
\binom{i}{\ell}
\widehat{\alpha}_{m-1,i}
\delta_m^{\,i-\ell},
\)
\(
\ell=0,\ldots,p_{\theta,m},
\label{eq:angle_coefficient_translation}
\)
and
\(
\beta_{m,\ell}^{\mathrm{init}}
=
\sum_{i=\ell}^{p_{r,m-1}}
\binom{i}{\ell}
\widehat{\beta}_{m-1,i}
\delta_m^{\,i-\ell},
\)
\(
\ell=0,\ldots,p_{r,m}.
\)
The translated vector is zero padded when the degree increases
and truncated when the degree decreases.
The closed-form gradient construction for the ULA and UPA models is summarized in Appendix~\ref{app:grad}.

\section{MLE-Based Payload-Aided TS Beam Tracking}
\label{sec:method}
This section develops two payload-beam selection strategies based on the MLE in \Cref{eq:mle_problem}: pure exploitation and uncertainty-aware TS.

\subsection{Pure Exploitation}\label{exploit}
Pure exploitation always transmits toward the channel predicted by the current MLE. During the $(m+1)$-th interval, it uses the estimate obtained at the end of the $m$-th interval:
\begin{equation}\label{eq:pillot-2}
\mathbf w_u^{\mathrm{PE}}
=
\left.
\widehat{\mathbf h}_{m,\mathrm{ULA}}
(t'_{u,m};\widehat{\boldsymbol\eta}_m)
\middle/
\left\|
\widehat{\mathbf h}_{m,\mathrm{ULA}}
(t'_{u,m};\widehat{\boldsymbol\eta}_m)
\right\|_2
\right. .
\end{equation}
Although this strategy preserves immediate beamforming gain, it yields limited new information about other potential directions. Consequently, the absence of exploration can lead to track loss once the predicted path diverges from the actual UE trajectory.

\subsection{MLE-Based Constant-Parameter TS Beam Tracking}\label{sec:TS}
Thompson sampling (TS) is a Bayesian randomized decision rule that draws a model instance from the current posterior and selects the action that is optimal for that instance \cite{russo2018tutorial}. Here, TS maps uncertainty in a continuous angle--range trajectory directly to payload-carrying beams; this use does not rely on a cumulative-regret characterization.

At the end of the $m$-th interval, we approximate the trajectory-parameter posterior by
\begin{equation}\label{eq:post_dist}
p(\boldsymbol\eta_m\mid\mathcal H_m)
\approx\mathcal N(\widehat{\boldsymbol\eta}_m,\boldsymbol\Sigma_{\boldsymbol\eta,m}),
\end{equation}
where the covariance is the inverse observed Fisher information matrix (FIM),
\begin{equation}\label{eq:fim}
\boldsymbol\Sigma_{\boldsymbol\eta,m}
\triangleq
\mathcal I_{\boldsymbol\eta,m}^{-1}
=
\left.
\left(
\frac{1}{\sigma^2}
\nabla_{\boldsymbol\eta}^2 J_m^{\mathrm{ULA}}(\boldsymbol\eta)
\right)^{-1}
\right|_{\boldsymbol\eta=\widehat{\boldsymbol\eta}_m}.
\end{equation}
This distribution provides a local uncertainty approximation around the selected MLE solution. Under the regularity and
posterior-concentration assumptions stated in Appendix~\ref{app:laplace_fim}, the
approximation error admits an
$O(|\mathcal U_m|^{-1})$ KL-divergence bound.

During the $(m+1)$-th interval, this posterior is held fixed; independent samples are drawn only at its designated feedback positions. For $u=u(m+1,k)$ with $k\in\mathcal K_{F,m+1}$, draw $\widetilde{\boldsymbol\eta}_{m,k}\sim\mathcal N(\widehat{\boldsymbol\eta}_m,\boldsymbol\Sigma_{\boldsymbol\eta,m})$ and define
$\widetilde{\mathbf h}^{(m)}_u\triangleq\widehat{\mathbf h}_{m,\mathrm{ULA}}(t'_{u,m};\widetilde{\boldsymbol\eta}_{m,k})$. The TS beam is
{
\setlength{\abovedisplayskip}{2pt}
\setlength{\abovedisplayshortskip}{2pt}
\begin{equation}\label{eq:pilot-3}
\begin{aligned}
\mathbf w_u^{\mathrm{TS}}
&=\underset{\|\mathbf w\|_2=1}{\arg\max}\;
\left|\widetilde{\mathbf h}^{(m)\mathsf H}_u\mathbf w\right|^2
= \widetilde{\mathbf h}^{(m)}_u / \|\widetilde{\mathbf h}^{(m)}_u\|_2.
\end{aligned}
\end{equation}
}

The posterior covariance governs the exploration--exploitation tradeoff: sampling produces larger perturbations along high-variance eigendirections and thereby probes trajectory variations that are weakly identified by the current observations.
In the proposed protocol, feedback-designated payload symbols use \Cref{eq:pilot-3}; the other payload symbols use the pure-exploitation beam in \Cref{eq:pillot-2}. Combining this batched decision rule with fixed tracking parameters yields the MLE-based constant-parameter Thompson sampling (MLE-CTS) algorithm in \Cref{alg:ts_beam_tracking}. 

\begin{algorithm}[!t]
\caption{MLE-CTS Beam Tracking}
\label{alg:ts_beam_tracking}
\begin{algorithmic}[1]
\renewcommand{\algorithmicrequire}{\textbf{Input:}}
\renewcommand{\algorithmicensure}{\textbf{Output:}}

\REQUIRE Warm-up history $\mathcal H_0$, its MLE and covariance, number of intervals $M$, and fixed $T_H$, $\Delta T$, $\omega$, $p_\theta$, and $p_r$.
\FOR{$m\leftarrow0,1,\ldots,M-1$}
    \STATE Compute $K_{m+1}$ and $K_{F,m+1}$ from their definitions, and construct $\mathcal K_{F,m+1}$ using \eqref{eq:feedback_set}.
    \STATE Initialize the new feedback-tuple set $\mathcal D_{m+1}\leftarrow\varnothing$.
    \FOR{$k\leftarrow1,2,\ldots,K_{m+1}$}
        \STATE Set $u\leftarrow u(m+1,k)$ and $t'_{u,m}\leftarrow t_u-(t_m-T_H)$.
        \IF{$k\in\mathcal K_{F,m+1}$}
            \STATE Draw $\widetilde{\boldsymbol\eta}_{m,k}$ from \eqref{eq:post_dist} and construct $\mathbf w_u$ using \eqref{eq:pilot-3}.
        \ELSE
            \STATE Construct $\mathbf w_u$ using \eqref{eq:pillot-2}.
        \ENDIF
        \STATE Transmit payload symbol $x_u$; if $k\in\mathcal K_{F,m+1}$, receive the feedback sample $y_u$ and append $(t_u,x_u,\mathbf w_u,y_u)$ to $\mathcal D_{m+1}$.
    \ENDFOR
    \STATE $\mathcal H_{m+1}\leftarrow\mathrm{Truncate}(\mathcal H_m\cup\mathcal D_{m+1},T_H)$.
    \STATE Solve \eqref{eq:mle_problem} for $\widehat{\boldsymbol\eta}_{m+1}$ and update \Cref{eq:fim} for the next interval.
\ENDFOR

\end{algorithmic}
\end{algorithm}

\section{Adaptive Tracking-Parameter Adjustment}
\label{sec:adaptive}
The MLE-CTS scheme in \Cref{alg:ts_beam_tracking} uses fixed tracking parameters, including the observation-window length, trajectory-model dimension, tracking-update interval, and feedback ratio. Fixed parameters are effective under smooth mobility but can become mismatched after abrupt changes in the local trajectory. A longer observation window supplies more feedback observations but may span incompatible motion regimes after a sharp turn. A higher polynomial degree increases local modeling flexibility at the cost of estimation variance, poorer numerical conditioning, and computational complexity. Likewise, shorter update intervals and higher feedback ratios improve responsiveness and information acquisition at the cost of computation and feedback overhead.

The adaptive adjustment is performed at the end of each update
interval. During interval $m$, the payload beams are constructed using the trajectory estimate and posterior obtained at the end of interval $m-1$. After the feedback from interval $m$ has been collected, three observation-window and trajectory-model configurations are evaluated. The selected configuration produces
$\widehat{\boldsymbol\eta}_m^\star$
and
$\boldsymbol\Sigma_{\boldsymbol\eta,m}^\star$,
which are used to construct the payload beams during interval $m+1$. The same update also determines
$\Delta T_{m+1}$ and $\omega_{m+1}$.
The update interval $\Delta T_{m+1}$ determines how long the selected trajectory estimate is used before the next MLE update, whereas the feedback ratio $\omega_{m+1}$ determines the number and positions of the received payload samples returned by the UE. After the feedback from interval $m+1$ has been collected, the observation-window length and trajectory-model dimension are selected again.

To ensure that an enlarged observation-window candidate remains
available after the nominal window has been shortened, MLE-ATS
maintains an auxiliary feedback history covering the maximum
admissible window length. Define the maximal feedback-index set as
$\mathcal U_m^{\max}\triangleq
\{u:t_m-T_{H,\max}\le t_u\le t_m,\;
u\text{ indexes a feedback sample}\}$,
and the corresponding maximal history as
$\mathcal H_m^{\max}\triangleq
\{(t_u,x_u,\mathbf w_u,y_u)\}_{u\in\mathcal U_m^{\max}}$.
The nominal fitting history $\mathcal H_m$ is the subset of
$\mathcal H_m^{\max}$ corresponding to the current window
$T_{H,m}$. All candidate-specific datasets are extracted from
$\mathcal H_m^{\max}$.
Throughout this section, $p_m$ denotes the total number of
real-valued coefficients in the local angle--range trajectory model. Let $\mathcal P$ denote the admissible total coefficient dimensions; in the simulations, $\mathcal P=\{10,11,12,13\}$, $T_{H,\min}=1\,\mathrm{ms}$, and $T_{H,\max}=120\,\mathrm{ms}$.

\subsection{Estimation-Risk-Guided Observation-Window and
Model-Dimension Selection}
At the end of interval $m$, we consider three candidate observation
windows: shortened, nominal, and enlarged. Let
\( \mathcal C \triangleq \{-1,0,+1\} \)
denote the three candidate indices. For $c\in\mathcal C$, define
\begin{equation}
T_{H,m}^{c}
\triangleq
\min
\left\{
T_{H,\max},
\max
\left\{
T_{H,\min},
T_{H,m}+c\varepsilon_T
\right\}
\right\},
\label{eq:candidate_window}
\end{equation}
where $\varepsilon_T\geq0$ controls the candidate-window adjustment.
The associated feedback-index set is
\begin{equation}
\mathcal U_m^c
\triangleq
\left\{
u\in\mathcal U_m^{\max}:
t_m-T_{H,m}^c
\le t_u\le t_m
\right\}.
\label{eq:candidate_feedback_set}
\end{equation}

For $u\in\mathcal U_m^c$, define the candidate-specific local time
\( t_{u,m}^{\prime,c}\triangleq t_u-(t_m-T_{H,m}^{c}) \)
and the corresponding predicted observation
\( \mu_u^c(\boldsymbol\eta)\triangleq
\widehat{\mathbf h}_{m,\mathrm{ULA}}^{\mathsf H}
(t_{u,m}^{\prime,c};\boldsymbol\eta)\mathbf w_ux_u. \)
For notational simplicity, the candidate-specific estimation risk
$R_{\mathcal U_m^c}$ below is evaluated using $\mu_u^c$.

For a generic feedback-index set $\mathcal S$, let
$\mu_u^\circ$ denote the noise-free received signal at index $u$:
\( y_u = \mu_u^\circ+n_u, \qquad n_u \sim \mathcal{CN}(0,\sigma^2). \)
For an arbitrary trajectory-parameter estimate
$\widehat{\boldsymbol\eta}$, define the estimation risk over
$\mathcal S$ as
\begin{equation}
R_{\mathcal S}
\left(
\widehat{\boldsymbol\eta}
\right)
\triangleq
\frac{1}{|\mathcal S|}
\sum_{u\in\mathcal S}
\left|
\mu_u
\left(
\widehat{\boldsymbol\eta}
\right)
-
\mu_u^\circ
\right|^2.
\label{eq:risk}
\end{equation}
This estimation risk measures the error in predicting the noise-free received
values along the transmitted beams. Because
$\{\mu_u^\circ\}_{u\in\mathcal S}$ are unknown at the BS,
\Cref{eq:risk} cannot be evaluated directly during tracking.

For a feedback-index set $\mathcal S$, let
$\widehat{\boldsymbol\eta}_{\mathcal S}$
denote the MLE obtained from the feedback observations indexed by
$\mathcal S$. The following lemma characterizes how the
estimation risk scales with the number of model
parameters and feedback observations.

\begin{lemma}
\label{lem:chi_square_risk}
Suppose that the local trajectory model is correctly specified over
$\mathcal S$ and contains $p$ real-valued parameters. Under the
smoothness, identifiability, and large-sample conditions stated in
Appendix~\ref{app:chi-square}, for fixed $p$,
\begin{equation}
\frac{
2|\mathcal S|
R_{\mathcal S}
\left(
\widehat{\boldsymbol\eta}_{\mathcal S}
\right)
}{
\sigma^2
}
\Rightarrow
\chi_p^2
\qquad
\text{as }
|\mathcal S|\rightarrow\infty,
\label{eq:risk_chi_square}
\end{equation}
where $\chi_p^2$ denotes the chi-square distribution with $p$
degrees of freedom. Under the additional moment condition stated in
Appendix~\ref{app:chi-square},
\begin{equation}
\mathbb E
\left[
R_{\mathcal S}
\left(
\widehat{\boldsymbol\eta}_{\mathcal S}
\right)
\right]
=
\frac{
\sigma^2p
}{
2|\mathcal S|
}
+
o
\left(
\frac{1}{|\mathcal S|}
\right).
\label{eq:risk_mean}
\end{equation}
\end{lemma}

\emph{Proof:}
See Appendix~\ref{app:chi-square}.

\textit{Remark 1:}
Lemma~1 evaluates the estimation risk on the same design points
used to fit the MLE. It therefore characterizes in-window
estimation-induced error, rather than extrapolation error over the
next update interval. We use \Cref{eq:risk_mean} only as a
model-complexity budget to coordinate the candidate model dimension with the number of available feedback observations. Candidate window lengths are compared using the
corrected estimation-risk score below.
The factor $2$ in \Cref{eq:risk_chi_square} arises because the real
and imaginary parts of
$n_u\sim\mathcal{CN}(0,\sigma^2)$
each have variance $\sigma^2/2$. 

For candidate $c$, define the estimation-risk-admissible model-dimension set
\begin{equation}
\mathcal P_m^c
\triangleq
\left\{
p\in\mathcal P:
\frac{\sigma^2p}{2|\mathcal U_m^c|}
\leq R_{\mathrm{dim}},
\quad
p<|\mathcal U_m^c|
\right\}.
\label{eq:pm_selection}
\end{equation}
If $\mathcal P_m^c=\varnothing$, no value is assigned to $p_m^c$,
and candidate $c$ is declared dimension-inadmissible.
Otherwise, set
$p_m^c\triangleq\max\mathcal P_m^c$
and map $p_m^c$ to its associated angle and range polynomial degrees.
For each dimension-admissible candidate, the trajectory coefficients
are estimated as
\begin{equation}
\widehat{\boldsymbol\eta}_m^{c}
=
\underset{
\boldsymbol\eta\in\Theta(p_m^{c})
}{
\arg\min}
\;
\sum_{u\in\mathcal U_m^{c}}
\left|
y_u-\mu_u^c(\boldsymbol\eta)
\right|^2,
\label{eq:candidate_mle}
\end{equation}
where $\Theta(p_m^{c})$ denotes the feasible parameter set associated
with model dimension $p_m^{c}$. For each dimension-admissible candidate for which
\Cref{eq:candidate_mle} is successfully solved, let
$\mathcal I_{\boldsymbol\eta,m}^{c}$
denote the candidate-specific observed FIM obtained by applying
\Cref{eq:fim} to the objective in
\Cref{eq:candidate_mle}.
Define the feasible-candidate set as $\mathcal C_m^{\mathrm{feas}} \triangleq \bigl\{c\in\mathcal C:\; \mathcal P_m^c\neq\varnothing,\; \widehat{\boldsymbol\eta}_m^c \text{ is successfully obtained}, \mathcal I_{\boldsymbol\eta,m}^{c} \text{ is positive definite} \bigr\}.$

To obtain an observable candidate-selection criterion, we relate the estimation risk to the received-signal residual. For an arbitrary estimate
$\widetilde{\boldsymbol\eta}$, define
\( e_u(\widetilde{\boldsymbol\eta})\triangleq y_u-\mu_u(\widetilde{\boldsymbol\eta}). \)

\begin{lemma}
\label{lem:residual_risk}
Let $\widetilde{\boldsymbol\eta}$ be fixed when the expectation over the noise samples $\{n_u:u\in\mathcal S\}$ is taken. Then
\begin{equation}
\mathbb E\left[\frac{1}{|\mathcal S|}\sum_{u\in\mathcal S}\left|e_u(\widetilde{\boldsymbol\eta})\right|^2\right]
=R_{\mathcal S}(\widetilde{\boldsymbol\eta})+\sigma^2.
\label{eq:residual_risk_decomp}
\end{equation}
\end{lemma}
\begin{proof}
Define $\Delta\mu_u\triangleq\mu_u(\widetilde{\boldsymbol\eta})-\mu_u^\circ$.
Since $y_u=\mu_u^\circ+n_u$, the residual satisfies
\(e_u(\widetilde{\boldsymbol\eta})=n_u-\Delta\mu_u\).
For fixed $\widetilde{\boldsymbol\eta}$, $\Delta\mu_u$ does not depend on $n_u$. Using $\mathbb E[n_u]=0$ and $\mathbb E[|n_u|^2]=\sigma^2$ gives
\(\mathbb E[|e_u(\widetilde{\boldsymbol\eta})|^2]=|\Delta\mu_u|^2+\sigma^2\).
Averaging over $u\in\mathcal S$ proves \Cref{eq:residual_risk_decomp}.
\end{proof}

Based on Lemma~\ref{lem:residual_risk}, define the observable noise-corrected estimation-risk score
\begin{equation}
\widehat R_{\mathcal S}(\widetilde{\boldsymbol\eta})
\triangleq
\frac{1}{|\mathcal S|}
\sum_{u\in\mathcal S}
\left|e_u(\widetilde{\boldsymbol\eta})\right|^2
-\sigma^2.
\label{eq:risk_eval}
\end{equation}

When $\widetilde{\boldsymbol\eta}$ is estimated from the same
observations used to evaluate the residual, the score
$\widehat R_{\mathcal S}(\widetilde{\boldsymbol\eta})$ is optimistically
biased because the fitted model absorbs part of the receiver
noise.  Appendix~C shows that, under the same
conditions as Lemma~1, the
same-sample first-order expansion gives
\begin{equation}
\mathbb E
\left[
\widehat R_{\mathcal S}
\left(
\widehat{\boldsymbol\eta}_{\mathcal S}
\right)
\right]
=
\mathbb E
\left[
R_{\mathcal S}
\left(
\widehat{\boldsymbol\eta}_{\mathcal S}
\right)
\right]
-
\frac{\sigma^2 p}{|\mathcal S|}
+
o
\left(
\frac{1}{|\mathcal S|}
\right).
\label{eq:same_sample_optimism}
\end{equation}
For each $c\in\mathcal C_m^{\mathrm{feas}}$, define
\(
e_u^c
\left(
\widetilde{\boldsymbol\eta}
\right)
\triangleq
y_u
-
\mu_u^c
\left(
\widetilde{\boldsymbol\eta}
\right).
\)
We use the following first-order optimism-corrected in-sample
estimation-risk score:
\begin{equation}
\widehat R_m^c
\triangleq
\frac{1}{
\left|
\mathcal U_m^c
\right|
}
\sum_{u\in\mathcal U_m^c}
\left|
e_u^c
\left(
\widehat{\boldsymbol\eta}_m^c
\right)
\right|^2
-
\sigma^2
+
\frac{
\sigma^2 p_m^c
}{
\left|
\mathcal U_m^c
\right|
}.
\label{eq:candidate_residual_score}
\end{equation}
Provided that
$\mathcal C_m^{\mathrm{feas}}\neq\varnothing$,
the selected candidate is
\begin{equation}
c_m^\star
=
\arg\min_{
c\in\mathcal C_m^{\mathrm{feas}}
}
\widehat R_m^c,
\label{eq:candidate_selection}
\end{equation}
If $\mathcal C_m^{\mathrm{feas}}=\varnothing$, no value is assigned to
$c_m^\star$, and the hold-last-valid-state branch in
\Cref{alg:adaptive_tracking} is invoked.


When $\mathcal C_m^{\mathrm{feas}}\neq\varnothing$, the selected
parameters are
\begin{equation}
T_{H,m}^{\star}=T_{H,m}^{c_m^\star},\quad
p_m^\star=p_m^{c_m^\star},\quad
\widehat{\boldsymbol\eta}_m^\star
=\widehat{\boldsymbol\eta}_m^{c_m^\star}.
\label{eq:selected_configuration}
\end{equation}
In this case, the posterior covariance
$\boldsymbol\Sigma_{\boldsymbol\eta,m}^\star$
is constructed from the selected estimate
$\widehat{\boldsymbol\eta}_m^\star$
using the observed-information matrix in \Cref{eq:fim}.
The estimate
$\widehat{\boldsymbol\eta}_m^\star$
and covariance
$\boldsymbol\Sigma_{\boldsymbol\eta,m}^\star$
are used to construct payload beams during the $(m+1)$-th interval. The
selected observation-window length and model dimension become the
nominal values used at the next update:
$T_{H,m+1} = T_{H,m}^\star$, $ p_{m+1} = p_m^\star. $
We map $p_{m+1}$ to the polynomial degrees as
$p_{\theta,m+1}=\lfloor(p_{m+1}-2)/2\rfloor$ and
$p_{r,m+1}=p_{m+1}-2-p_{\theta,m+1}$.

\subsection{Residual-Based Update-Interval and Feedback-Ratio
Adjustment}
\label{subsec:adaptive_interval_feedback}
When $\mathcal C_m^{\mathrm{feas}}\neq\varnothing$, the feedback
observations collected during interval $m$ are used to adjust the
update interval and feedback ratio for interval $m+1$. In this case, let
\(
\mathcal S_m^\star
\triangleq
\mathcal U_m^{c_m^\star}
\)
denote the fitting set associated with the selected candidate.
Let $\Delta\mathcal U_m$ denote the feedback indices collected
during interval $m$.
For any feedback index $u$, define
\(
e_{u,m}^\star
\triangleq
y_u
-
\mu_u^{c_m^\star}
\left(
\widehat{\boldsymbol\eta}_m^\star
\right)
\)
and
\(
\mathbf g_{u,m}^\star
\triangleq
\left.
\nabla_{\boldsymbol\eta}
\mu_u^{c_m^\star}
\left(
\boldsymbol\eta
\right)
\right|_{
\boldsymbol\eta=
\widehat{\boldsymbol\eta}_m^\star
}.
\)
For an index set $\mathcal B$, define the local design-information matrix
\begin{equation}
\mathbf I_m(\mathcal B)
\triangleq
\frac{2}{\sigma^2}
\sum_{u\in\mathcal B}
\operatorname{Re}
\left\{
\mathbf g_{u,m}^\star
(\mathbf g_{u,m}^{\star})^{\mathsf H}
\right\}.
\label{eq:subset_information}
\end{equation}
For any nonempty index set $\mathcal A$, define the effective
same-sample degrees of freedom
\begin{equation}
d_m(\mathcal A)
\triangleq
\operatorname{tr}
\left[
\mathbf I_m
\left(
\mathcal A\cap \mathcal S_m^\star
\right)
\mathbf I_m
\left(
\mathcal S_m^\star
\right)^{-1}
\right].
\label{eq:subset_effective_df}
\end{equation}
The corresponding first-order optimism-corrected estimation-risk score is
\begin{equation}
\widehat R_m^{\mathrm{ctrl}}(\mathcal A)
\triangleq
\frac{1}{|\mathcal A|}
\sum_{u\in\mathcal A}
\left|
e_{u,m}^\star
\right|^2
-
\sigma^2
+
\frac{\sigma^2d_m(\mathcal A)}{|\mathcal A|}.
\label{eq:controller_corrected_score}
\end{equation}

The intersection in
Eq.~\eqref{eq:subset_effective_df}
accounts for the fact that only observations used in fitting
$\widehat{\boldsymbol\eta}_m^\star$
produce same-sample optimism. In particular, if
$\mathcal A=\mathcal S_m^\star$, then
$d_m(\mathcal A)=p_m^\star$ and
Eq.~\eqref{eq:controller_corrected_score}
reduces to the correction used in
Eq.~\eqref{eq:candidate_residual_score}.

Let $\Delta\mathcal U_m^{\mathrm{head}}$ and $\Delta\mathcal U_m^{\mathrm{tail}}$ denote the feedback indices whose physical times lie in $[t_{m-1},(t_{m-1}+t_m)/2)$ and $[(t_{m-1}+t_m)/2,t_m)$, respectively. Define
\begin{equation}
R_m^{\mathrm{head}}
\triangleq
\widehat R_m^{\mathrm{ctrl}}
\left(
\Delta\mathcal U_m^{\mathrm{head}}
\right),
\qquad
R_m^{\mathrm{tail}}
\triangleq
\widehat R_m^{\mathrm{ctrl}}
\left(
\Delta\mathcal U_m^{\mathrm{tail}}
\right),
\label{eq:head_tail_risk}
\end{equation}
and
\begin{equation}
g_{\Delta,m}
\triangleq
R_m^{\mathrm{tail}}
-
R_m^{\mathrm{head}}.
\label{eq:risk_growth_indicator}
\end{equation}
A positive $g_{\Delta,m}$ indicates that the selected local model
has a larger estimation-risk score over the second half of interval
$m$, motivating a shorter update interval for interval $m+1$.

The corrected whole-interval estimation-risk score is
\begin{equation}
R_m^{\mathrm{whole}}
\triangleq
\widehat R_m^{\mathrm{ctrl}}
\left(
\Delta\mathcal U_m
\right),
\label{eq:whole_interval_risk}
\end{equation}
and the feedback-ratio indicator is
\begin{equation}
g_{\omega,m}
\triangleq
R_m^{\mathrm{whole}}
-
R_{\mathrm{fb}}.
\label{eq:feedback_risk_indicator}
\end{equation}
When $g_{\omega,m}>0$, the estimation-risk score exceeds the target
and the feedback ratio is increased for the $(m+1)$-th interval; when
$g_{\omega,m}<0$, the feedback ratio may be reduced.

The estimation-risk levels $R_{\mathrm{dim}}$ and $R_{\mathrm{fb}}$ have different roles. The former determines the candidate trajectory-model dimension through Lemma~\ref{lem:chi_square_risk}, whereas the latter is the feedback-control target for the observable estimation-risk score.
Let
$\rho_\Delta>0$
and
$\rho_\omega>0$
denote the two update gains. The tracking-update interval is adjusted
as
\begin{equation}
\begin{aligned}
\Delta T_{m+1}
&=
\min
\left\{
\Delta T_{\max},
\right.\\[-1mm]
&\hspace{8mm}\left.
\max
\left\{
\Delta T_{\min},
\Delta T_m
\exp
\left(
-\rho_\Delta g_{\Delta,m}
\right)
\right\}
\right\},
\end{aligned}
\label{eq:update_interval_rule}
\end{equation}
and the feedback ratio is adjusted as
\begin{equation}
\omega_{m+1}
=
\min
\left\{
\omega_{\max},
\max
\left\{
\omega_{\min},
\omega_m
\exp
\left(
\rho_\omega g_{\omega,m}
\right)
\right\}
\right\}.
\label{eq:feedback_ratio_rule}
\end{equation}

In the simulations, the allowed ranges are
$
\Delta T_{\min}
=
1\,\mathrm{ms},
$
$
\Delta T_{\max}
=
40\,\mathrm{ms},
$
$
\omega_{\min}
=
0.5,
$
$
\omega_{\max}
=
1.
$ The number of payload symbols in interval $m+1$ is
\begin{equation}
K_{m+1}
=
\max
\left\{
2,
\left\lfloor
\frac{\Delta T_{m+1}}{T_s}
\right\rfloor
\right\},
\qquad
\Delta T_{m+1}
\leftarrow
K_{m+1}T_s.
\label{eq:integer_interval_length}
\end{equation}
The number of feedback observations is
\begin{equation}
K_{F,m+1}
=
\min
\left\{
K_{m+1},
\max
\left\{
2,
\left\lceil
\omega_{m+1}K_{m+1}
\right\rceil
\right\}
\right\}.
\label{eq:integer_feedback_number}
\end{equation}
The realized feedback ratio can then be set to
\( \omega_{m+1} \leftarrow \frac{ K_{F,m+1} }{ K_{m+1} }. \)
The feedback-position set for interval $m+1$ is constructed using
\Cref{eq:feedback_set}. The estimate
$\widehat{\boldsymbol\eta}_m^\star$
and posterior covariance
$\boldsymbol\Sigma_{\boldsymbol\eta,m}^\star$
determine the payload-beam directions during interval $m+1$.
The feedback ratio
$\omega_{m+1}$
determines the TS probing and receiver-feedback positions, whereas
$\Delta T_{m+1}$
determines the duration of interval $m+1$ and the time of the next
MLE update. The observation-window length and trajectory-model
dimension are selected again when that next update is performed. 
After the feedback from the $(m+1)$-th interval has been collected, the
maximal history is updated as
\begin{equation}
\mathcal H_{m+1}^{\max}
\leftarrow
\operatorname{Truncate}
\left(
\mathcal H_m^{\max}
\cup
\mathcal D_{m+1},
T_{H,\max}
\right).
\label{eq:maximal_history_update}
\end{equation}
The history is not truncated to the selected nominal window
$T_{H,m+1}$, because the enlarged candidate at the next update may
require older observations. When no feasible candidate is available, the current adaptive update
is skipped. The most recent valid trajectory posterior is retained
together with its original local-time origin, and the current nominal
values $T_{H,m}$, $p_m$, $\Delta T_m$, and $\omega_m$ are carried over
to interval $m+1$.
The complete procedure is summarized in
\Cref{alg:adaptive_tracking}.

\begin{algorithm}[!t]
\caption{Adaptive Tracking-Parameter Adjustment for MLE-ATS}
\label{alg:adaptive_tracking}
\begin{algorithmic}[1]

\renewcommand{\algorithmicrequire}{\textbf{Input:}}
\renewcommand{\algorithmicensure}{\textbf{Output:}}

\REQUIRE
Most recent valid posterior state;
maximal history $\mathcal H_m^{\max}$ with index set
$\mathcal U_m^{\max}$;
feedback indices $\Delta\mathcal U_m$ collected during interval $m$;
current nominal values
$T_{H,m}$, $p_m$, $\Delta T_m$, and $\omega_m$;
noise variance $\sigma^2$;
window increment $\varepsilon_T$;
admissible model dimensions $\mathcal P$;
    estimation-risk levels $R_{\mathrm{dim}}$ and $R_{\mathrm{fb}}$;
update gains $\rho_\Delta,\rho_\omega>0$.

\ENSURE
$\widehat{\boldsymbol\eta}_m^\star$,
$\boldsymbol\Sigma_{\boldsymbol\eta,m}^\star$,
$T_{H,m+1}$,
$p_{m+1}$,
$\Delta T_{m+1}$,
$\omega_{m+1}$,
$K_{m+1}$, and
$K_{F,m+1}$.

\FOR{$c\in\mathcal C=\{-1,0,+1\}$}

    \STATE Construct $T_{H,m}^{c}$, $\mathcal U_m^c$, and
    $t_{u,m}^{\prime,c}$ using
    \Cref{eq:candidate_window,eq:candidate_feedback_set},
    and construct $\mathcal P_m^c$ using \Cref{eq:pm_selection}.

    \IF{$\mathcal P_m^c\neq\varnothing$}

        \STATE Set $p_m^c\leftarrow\max\mathcal P_m^c$ and attempt to
        compute $\widehat{\boldsymbol\eta}_m^c$ using
        \Cref{eq:candidate_mle}.

        \IF{the MLE is successfully obtained and its observed FIM is positive definite}

            \STATE Compute $\widehat R_m^c$ using
            \Cref{eq:candidate_residual_score}.

        \ENDIF

    \ENDIF

\ENDFOR

\IF{$\mathcal C_m^{\mathrm{feas}}=\varnothing$}

    \STATE Retain the most recent valid posterior state; set
    $T_{H,m+1}\leftarrow T_{H,m}$,
    $p_{m+1}\leftarrow p_m$,
    $\Delta T_{m+1}\leftarrow\Delta T_m$, and
    $\omega_{m+1}\leftarrow\omega_m$.

\ELSE

    \STATE Select $c_m^\star$ using \Cref{eq:candidate_selection},
    and set $T_{H,m}^\star$, $p_m^\star$, and
    $\widehat{\boldsymbol\eta}_m^\star$
    using \Cref{eq:selected_configuration}.

    \STATE Construct
    $\boldsymbol\Sigma_{\boldsymbol\eta,m}^\star$
    using \Cref{eq:fim}; set
    $T_{H,m+1}=T_{H,m}^\star$ and $p_{m+1}=p_m^\star$.

    \STATE Split $\Delta\mathcal U_m$ into
    $\Delta\mathcal U_m^{\mathrm{head}}$ and
    $\Delta\mathcal U_m^{\mathrm{tail}}$, and compute
    $R_m^{\mathrm{head}}$, $R_m^{\mathrm{tail}}$, and
    $R_m^{\mathrm{whole}}$ using
    \Cref{eq:head_tail_risk,eq:whole_interval_risk}.

    \STATE Compute $g_{\Delta,m}$ and $g_{\omega,m}$ and update
    $\Delta T_{m+1}$ and $\omega_{m+1}$ using
    \Cref{eq:risk_growth_indicator,eq:feedback_risk_indicator,eq:update_interval_rule,eq:feedback_ratio_rule}.

\ENDIF

\STATE Compute $K_{m+1}$ and $K_{F,m+1}$ using
\Cref{eq:integer_interval_length,eq:integer_feedback_number},
and construct $\mathcal K_{F,m+1}$ using \Cref{eq:feedback_set}.

\end{algorithmic}
\end{algorithm}

Combining the payload-beam selection in
\Cref{alg:ts_beam_tracking}
with the adaptive tracking-parameter adjustment in
\Cref{alg:adaptive_tracking}
gives the
\emph{MLE-based adaptive-parameter Thompson sampling}
scheme, abbreviated as \emph{MLE-ATS}.

\section{Simulation Results}\label{sec:sim}

The simulation uses a 256-element ULA at $f_c=73$ GHz. The total duration is $T=4$ s and the symbol duration is $T_s=1/(30\,\mathrm{kHz})$. The channel contains one time-varying LoS path and one path reflected by a stationary scatterer \cite{NYU_NLOS}. The UE remains in the near-field region, with $r\in[8,80]$ m and $\theta\in[-\pi/3,\pi/3]$, and follows a time-varying velocity profile. For each target initial SNR: $\rho_{\mathrm{dB}}\equiv\mathrm{SNR}_{0,\mathrm{dB}}$, we define $\rho\triangleq10^{\rho_{\mathrm{dB}}/10}=\mathrm{SNR}_0$, set $\sigma^2=\|\mathbf h_{0,\mathrm{ULA}}\|_2^2/(N\rho)$, and hold $\sigma^2$ fixed throughout the run. For the MLVS baseline \cite{tvt2024tracking}, we use the numerical-example setting $\sigma_{\mathrm{RCS}}=-23$ dBsm.
 
\subsection{Performance Metrics}
Let $\chi_u=1$ when the $u$th symbol carries payload and $\chi_u=0$ for a dedicated pilot. With $G_u$ defined by \Cref{eq:normalized_gain}, the payload-accounted normalized beamforming gain is $\widetilde G_u\triangleq\chi_uG_u$; hence, dedicated pilots contribute zero payload gain. Over an evaluation set $\mathcal V$, the mean and variance are $\overline G=|\mathcal V|^{-1}\sum_{u\in\mathcal V}\widetilde G_u$ and $V_G=|\mathcal V|^{-1}\sum_{u\in\mathcal V}(\widetilde G_u-\overline G)^2$, respectively.
We set the normalized-gain threshold to $\gamma=0.75$ and count a transmitted symbol as effective only if it carries payload and satisfies $G_u\geq\gamma$. Equivalently, the effective-symbol ratio over all transmitted symbols in $\mathcal V$ is
$ P_{\mathrm{eff}}(\gamma)=\frac{1}{|\mathcal V|} \sum_{u\in\mathcal V}\mathbbm 1\{\widetilde G_u \geq \gamma\}$. 



\subsection{Evaluation of MLE-CTS Beam Tracking}
We first use a parabolic trajectory to evaluate MLE-CTS against pure exploitation and the representative baselines. The UE's mean and maximum speeds are $3.11$ m/s and $4.71$ m/s, respectively. The polynomial model uses $p_\theta=3$ and $p_r=6$ with $T_H=66.60$ ms. 
MLE-ATS is also evaluated under the parabolic trajectory.
When MLE-ATS selects the total model dimension $p_{m+1}$,
we keep the angular polynomial degree fixed at
$p_{\theta,m+1}=3$ and assign the remaining coefficients to
the range model, i.e.,
$p_{r,m+1}=p_{m+1}-2-p_{\theta,m+1}$.
Because MLE-ATS and MLE-CTS achieve nearly identical tracking
performance in this smooth-mobility setting,
\Cref{fig:tracking_performance,fig:smooth_vs_snr} reports only the MLE-CTS results
for clarity. Unless otherwise stated, SNR is $20$ dB. 

\begin{figure}[!t]
    \centering
    \subfloat[Effect of $\Delta T$ on MLE-CTS and pure exploitation]{
        \begin{minipage}{0.46\linewidth}
            \centering
            \includegraphics[width=\linewidth]{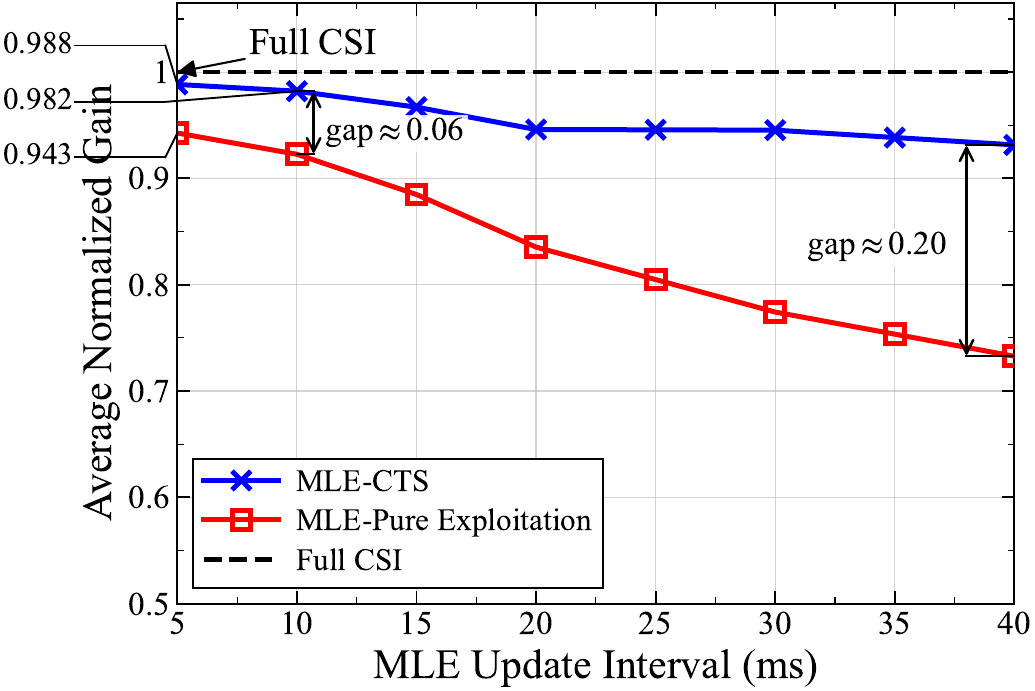}
            \vspace{-5mm}
            \label{fig:normalized beamforming gain}
        \end{minipage}
    }
    \hspace{0.01\linewidth}
    \subfloat[Effect of feedback ratio on MLE-CTS ($\Delta T=10$ ms)]{
        \begin{minipage}{0.46\linewidth}
            \centering
            \includegraphics[width=\linewidth]{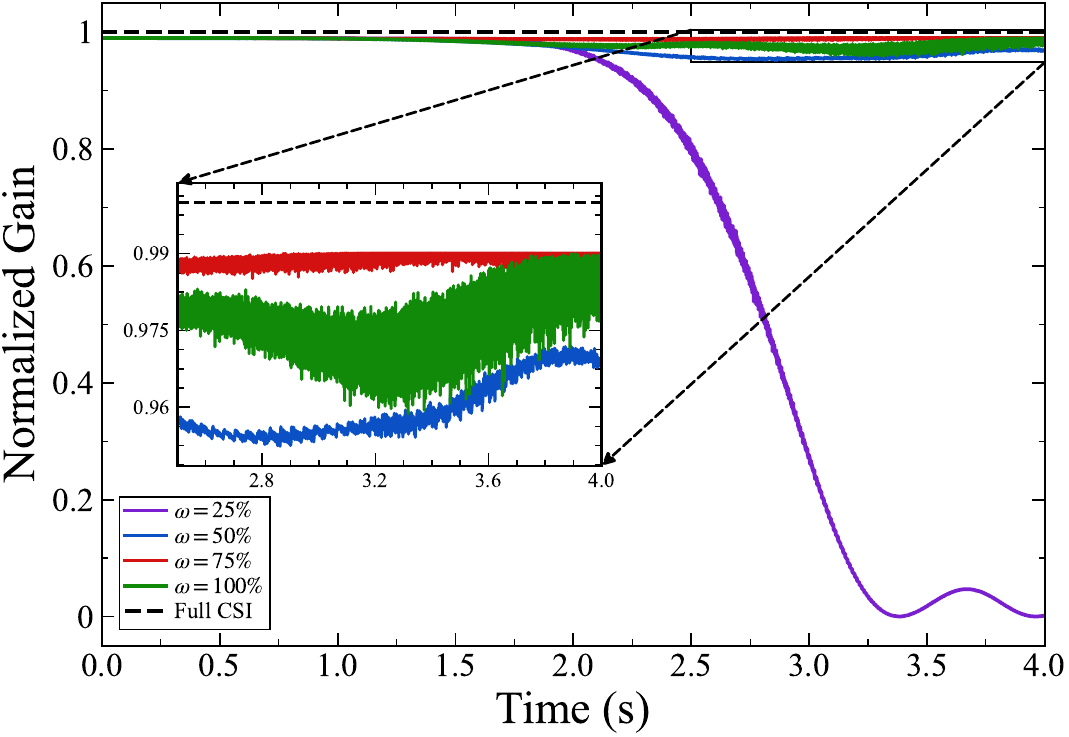}
            \vspace{-5mm}
            \label{fig:feedback}
        \end{minipage}
    }
    \vspace{2mm}
    \caption{Tracking performance comparison under varying MLE update intervals and feedback symbol ratios.}
    \label{fig:tracking_performance}
\end{figure}


\Cref{fig:normalized beamforming gain} compares the proposed MLE-CTS tracking scheme with pure exploitation under different MLE update intervals $\Delta T$, using the mean payload-accounted normalized beamforming gain over $T$ as the performance metric. As $\Delta T$ increases, the tracking performance of both schemes deteriorates because the beamformer relies on increasingly outdated motion estimates between consecutive updates. 
Nevertheless, the proposed MLE-CTS consistently outperforms pure exploitation over the entire trajectory for the same $\Delta T$.
When the update interval is small, such as $5$ ms, MLE-CTS achieves a mean normalized beamforming gain of $0.99$, which approaches the theoretical upper bound because the estimated trajectory parameters, $\hat{\theta}(t)$ and $\hat{r}(t)$, are frequently refined. Conversely, pure exploitation degrades as tracking progresses because it cannot acquire new channel information through alternative probing beams. As $\Delta T$ increases, the performance gap between the two schemes becomes more pronounced. Specifically, when $\Delta T = 10$ ms, the difference in mean beamforming gain is about $0.06$, whereas it grows to approximately $0.20$ as $\Delta T$ increases to $40$ ms. The widening gap indicates that posterior-guided exploration makes MLE-CTS more robust to accumulated prediction errors than pure exploitation.

The proposed payload-aided method requires closed-loop receiver feedback. To evaluate this overhead, we vary the feedback-symbol ratio $\omega$ within each update interval and set $\Delta T=10$ ms.
This interval was selected as a representative setting because, as shown in \Cref{fig:normalized beamforming gain}, the performance gap between $\Delta T = 5$ ms and $\Delta T = 10$ ms is negligible (approximately $0.01$).
As illustrated in \Cref{fig:feedback}, the proposed method remains effective even when $\omega$ is reduced: with only $50\%$ feedback symbols, the achieved normalized beamforming gain is close to that of the $100\%$-feedback strategy. 
Interestingly, the $75\%$ feedback strategy slightly outperforms the $100\%$ strategy. This occurs because the $100\%$ strategy continuously allocates symbols to exploration. In contrast, the $75\%$ strategy reserves a portion of the frame for pure exploitation. Once the channel estimate is sufficiently accurate, this greedy allocation avoids unnecessary perturbations from TS exploration, leading to more stable beamforming. Conversely, at a $25\%$ feedback ratio, the algorithm tracks the UE for only the first $2$ s because too few samples are available for accurate estimation.

With fixed $\Delta T$, $T_H$, and $p$, we compare MLE-CTS under the parabolic trajectory with three baselines: coherence-time-driven local sweeping \cite{UPA_NF}, maximum-likelihood velocity-sensing (MLVS) assisted tracking \cite{tvt2024tracking}, and extended Kalman filter (EKF) tracking. The coherence-time scheme derives an effective coherence time from a tolerable gain loss and predicted velocity, then performs a local near-field sweep around the predicted position. MLVS estimates radial and transverse velocities from echo observations and predicts the next angle--range state. The EKF models range, angle, and velocity as a nonlinear state, applies a constant-velocity prediction, and updates that prediction using feedback-pilot observations.







\begin{figure*}[!t]
    \centering
    \subfloat[Mean normalized beamforming gain.
    \label{fig:avg_gain_smooth}]{
        \includegraphics[width=0.30\textwidth]
        {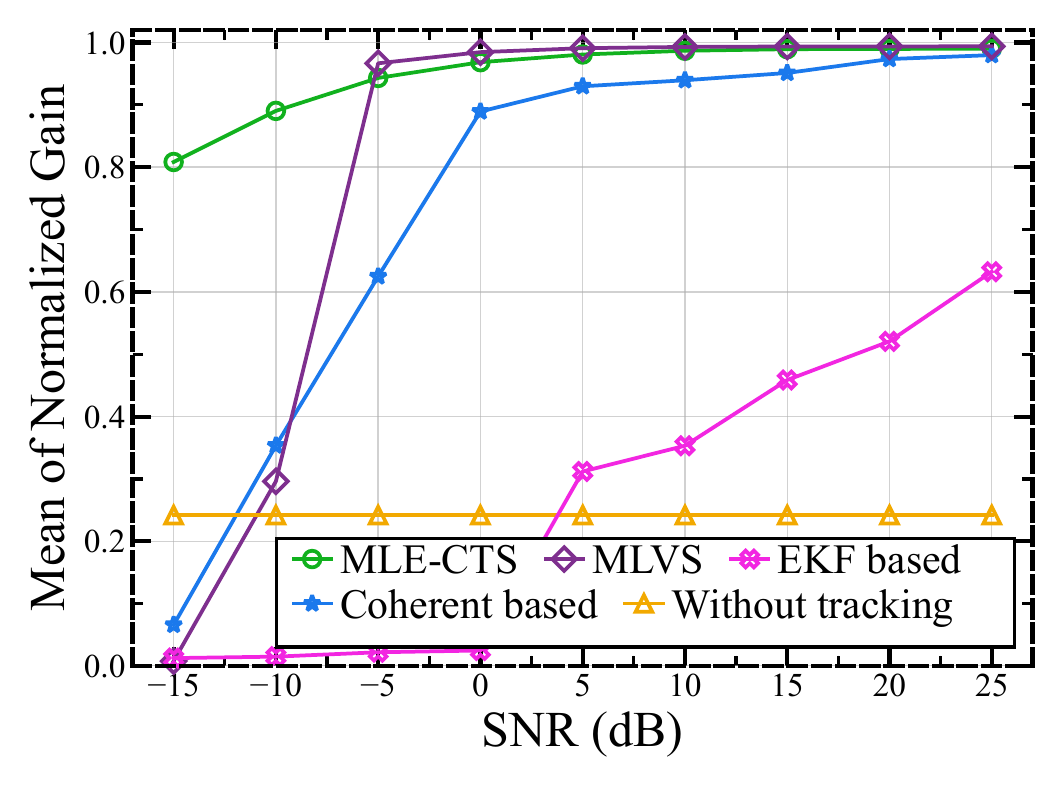}
    }%
    \hfill
    \subfloat[Payload-accounted normalized-gain variance.
    \label{fig:var_smooth}]{
        \includegraphics[width=0.31\textwidth]
        {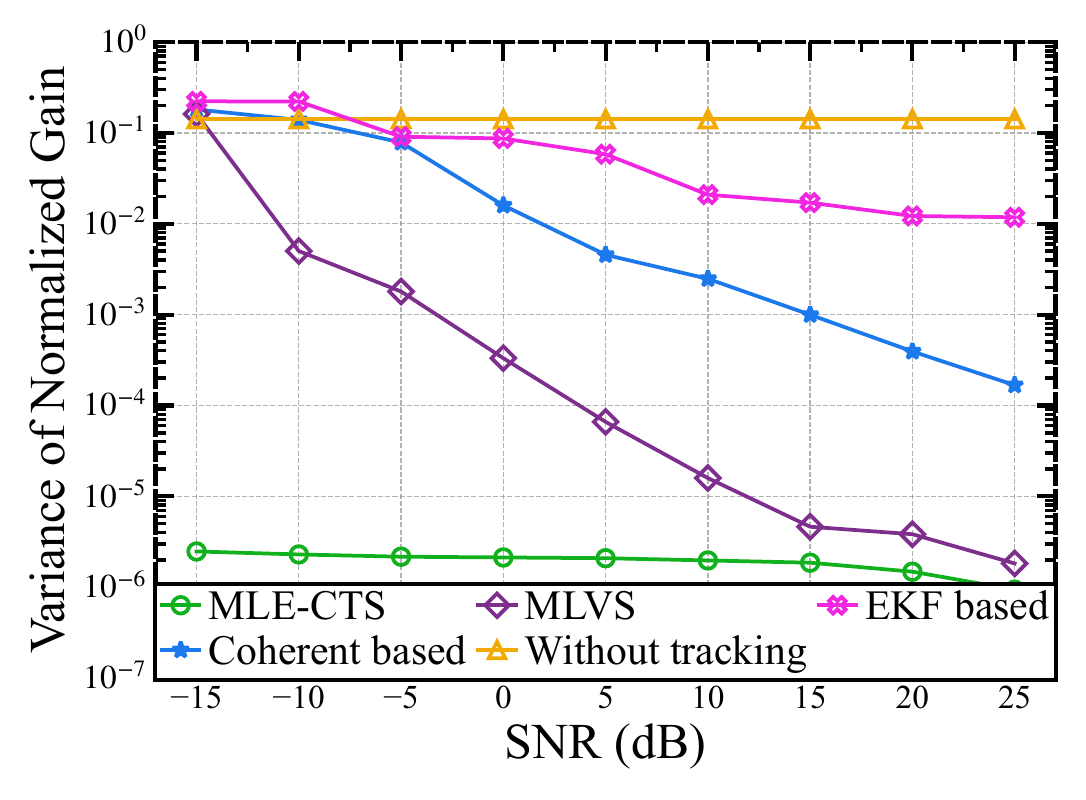}
    }%
    \hfill
    \subfloat[Effective-symbol ratio.
    \label{fig:outage_smooth}]{
        \includegraphics[width=0.30\textwidth]
        {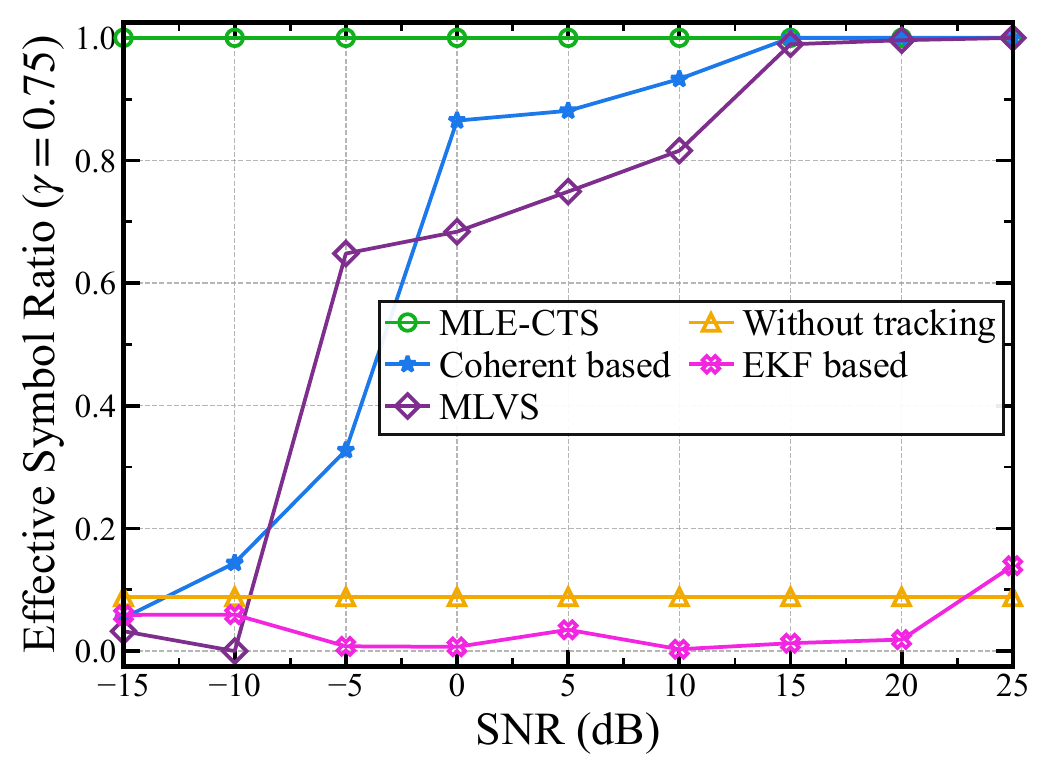}
    }
    \caption{Payload-accounted tracking performance versus SNR under
    the parabolic trajectory.}
    \label{fig:smooth_vs_snr}
\end{figure*}


\begin{figure}[t]
    \centering
    \includegraphics[width=0.92\linewidth]{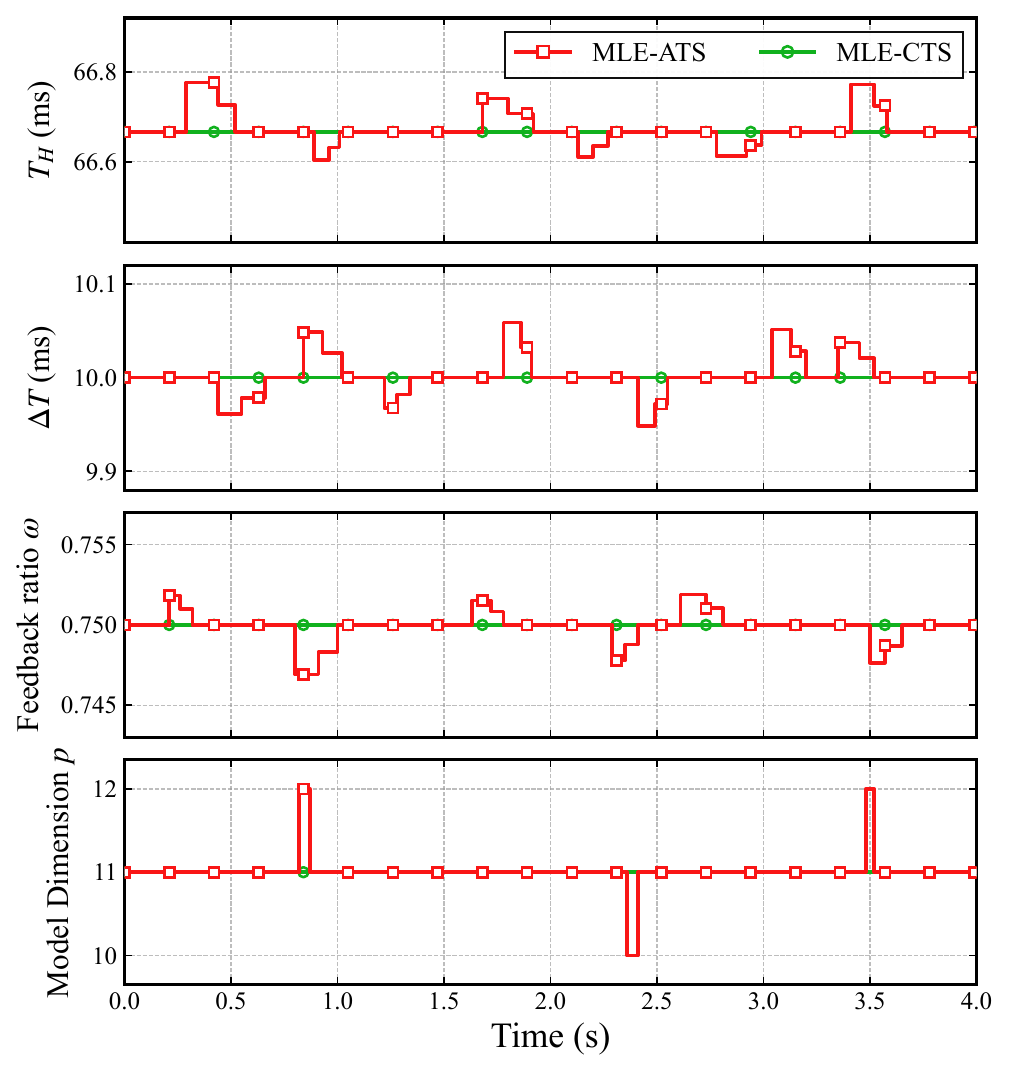}
    \caption{Tracking parameters of MLE-ATS and MLE-CTS at $20$ dB under the parabolic trajectory.}
    \label{fig:parameters_smooth}
\end{figure}

\Cref{fig:avg_gain_smooth} shows the mean normalized beamforming gain versus SNR under the parabolic trajectory. MLE-CTS yields the highest mean gain at the two lowest SNRs and approaches unity as SNR increases. At $-15$ dB, it achieves approximately $0.81$, whereas the coherence-time, MLVS, and EKF baselines are close to zero and the no-tracking benchmark remains near $0.24$. MLVS marginally exceeds MLE-CTS from $-5$ dB onward after its sensing estimate stabilizes, while the coherence-time method gradually closes the gap as SNR increases.

\Cref{fig:var_smooth} shows that MLE-CTS has the lowest payload-accounted normalized-gain variance, remaining on the order of $10^{-6}$ across the SNR sweep. The variance of MLVS decreases from approximately $10^{-1}$ at low SNR to the $10^{-6}$ range at high SNR. The coherence-time method exhibits a similar SNR-dependent reduction but retains a larger variance because local beam sweeping introduces zero payload-gain entries, while the EKF and no-tracking baselines retain substantially larger values.

\Cref{fig:outage_smooth} shows that MLE-CTS maintains an effective-symbol ratio of $1.00$ throughout the SNR sweep, consistent with its low payload-accounted normalized-gain variance. At $0$ dB, MLVS and the coherence-time baseline achieve $0.85$ and $0.72$, respectively; at $25$ dB, they reach $0.90$ and $0.88$. The EKF and no-tracking baselines remain below $0.15$. Thus, MLE-CTS provides reliable per-symbol gain even where the competing methods exhibit comparable mean gain only at higher SNR.

\Cref{fig:parameters_smooth} compares the four parameters $T_H$, $\Delta T$, $\omega$, and $p$ for MLE-ATS and MLE-CTS under the parabolic trajectory. Starting from the same configuration, MLE-ATS introduces only minor variations, indicating that substantial adaptation is unnecessary when the mobility remains smooth.



\subsection{Evaluation of MLE-ATS Beam Tracking}

\begin{figure}[t]
    \centering
    \includegraphics[width=0.8\linewidth]{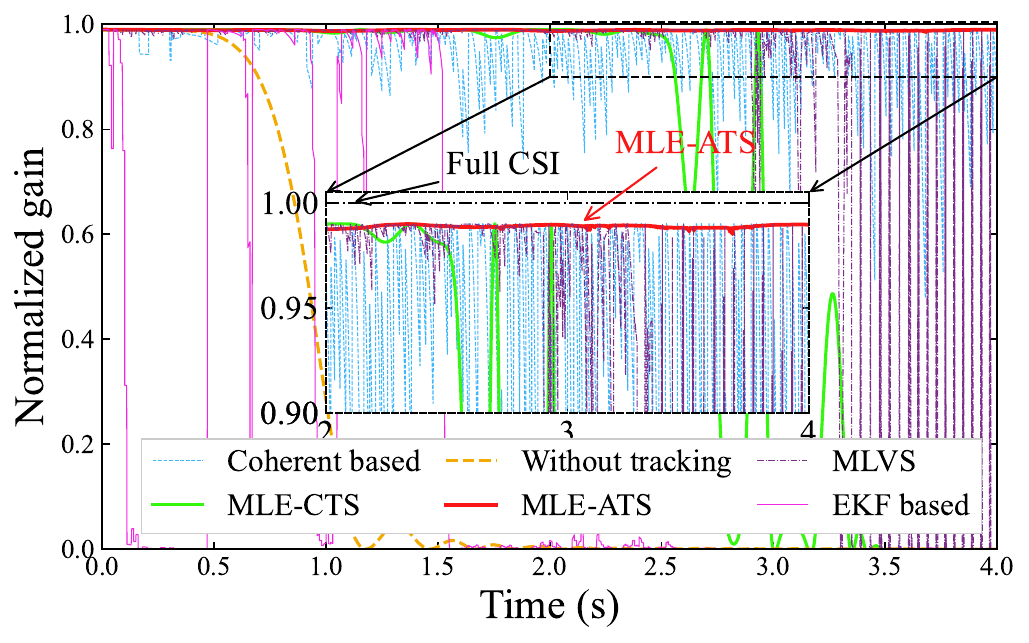}
    \caption{Normalized beamforming gain over time under the sharp-turn trajectory.}
    \label{fig:gain_sharp}
\end{figure}

\begin{figure}[t]
    \centering
    \includegraphics[width=0.8\linewidth]{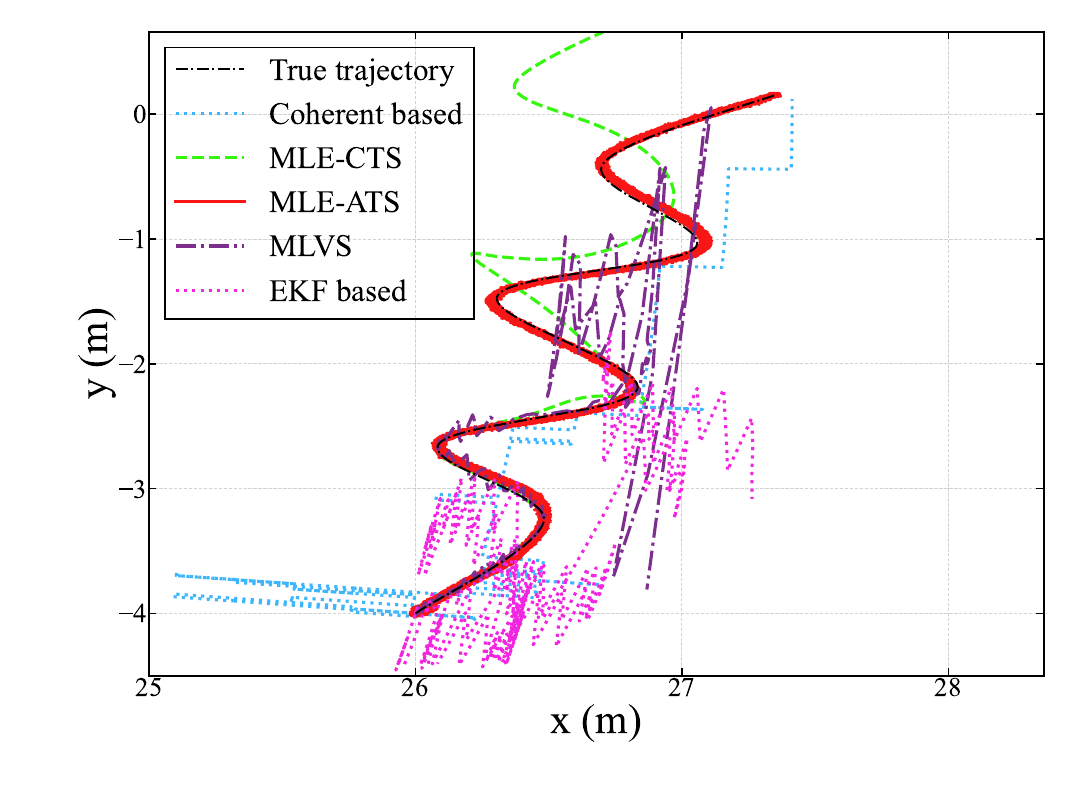}
    \caption{Trajectory-tracking comparison under the sharp-turn trajectory.}
    \label{fig:traj_sharp}
\end{figure}

\begin{figure*}[!t]
    \centering
    \subfloat[Mean normalized beamforming gain.\label{fig:avg_gain_sharp}]{
        \includegraphics[width=0.30\textwidth]
        {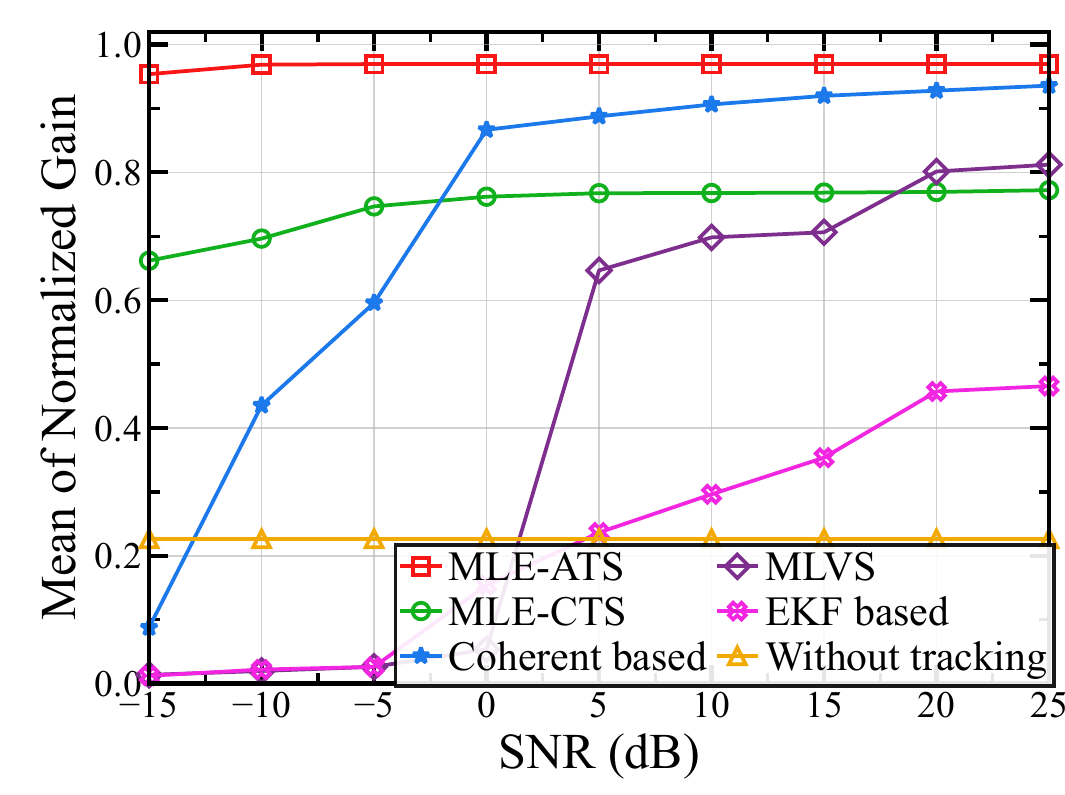}
    }%
    \hfill
    \subfloat[Payload-accounted normalized-gain variance.\label{fig:sharp_var_snr}]{
        \includegraphics[width=0.31\textwidth]
        {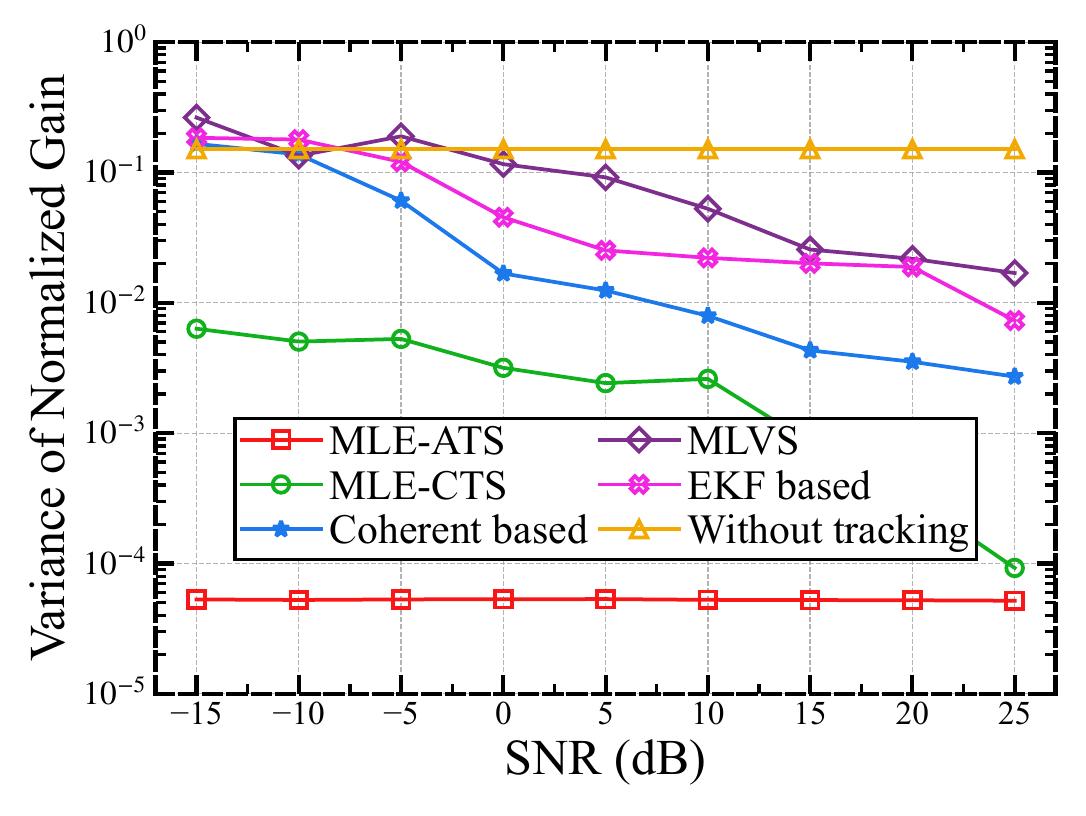}
    }%
    \hfill
    \subfloat[Effective-symbol ratio.\label{fig:outage_sharp}]{
        \includegraphics[width=0.30\textwidth]
        {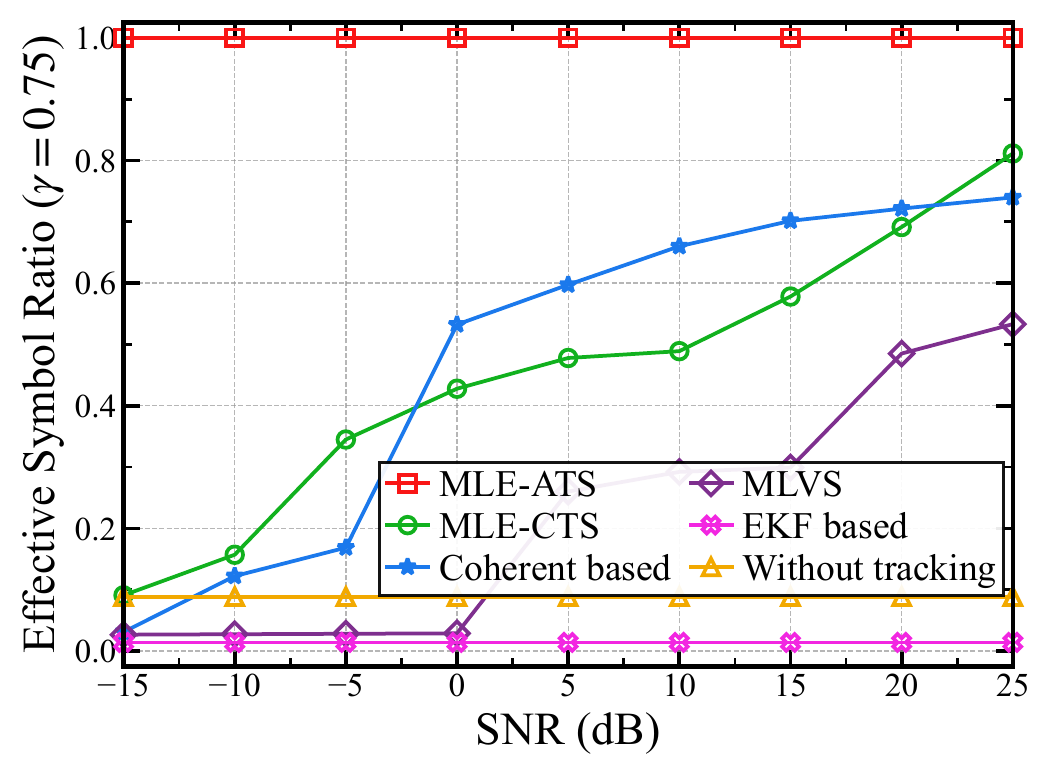}
    }
    \caption{Payload-accounted tracking performance versus SNR under
    the sharp-turn trajectory.}
    \label{fig:sharp_vs_snr}
\end{figure*}

\begin{figure}[t]
    \centering
    \includegraphics[width=0.92\linewidth]{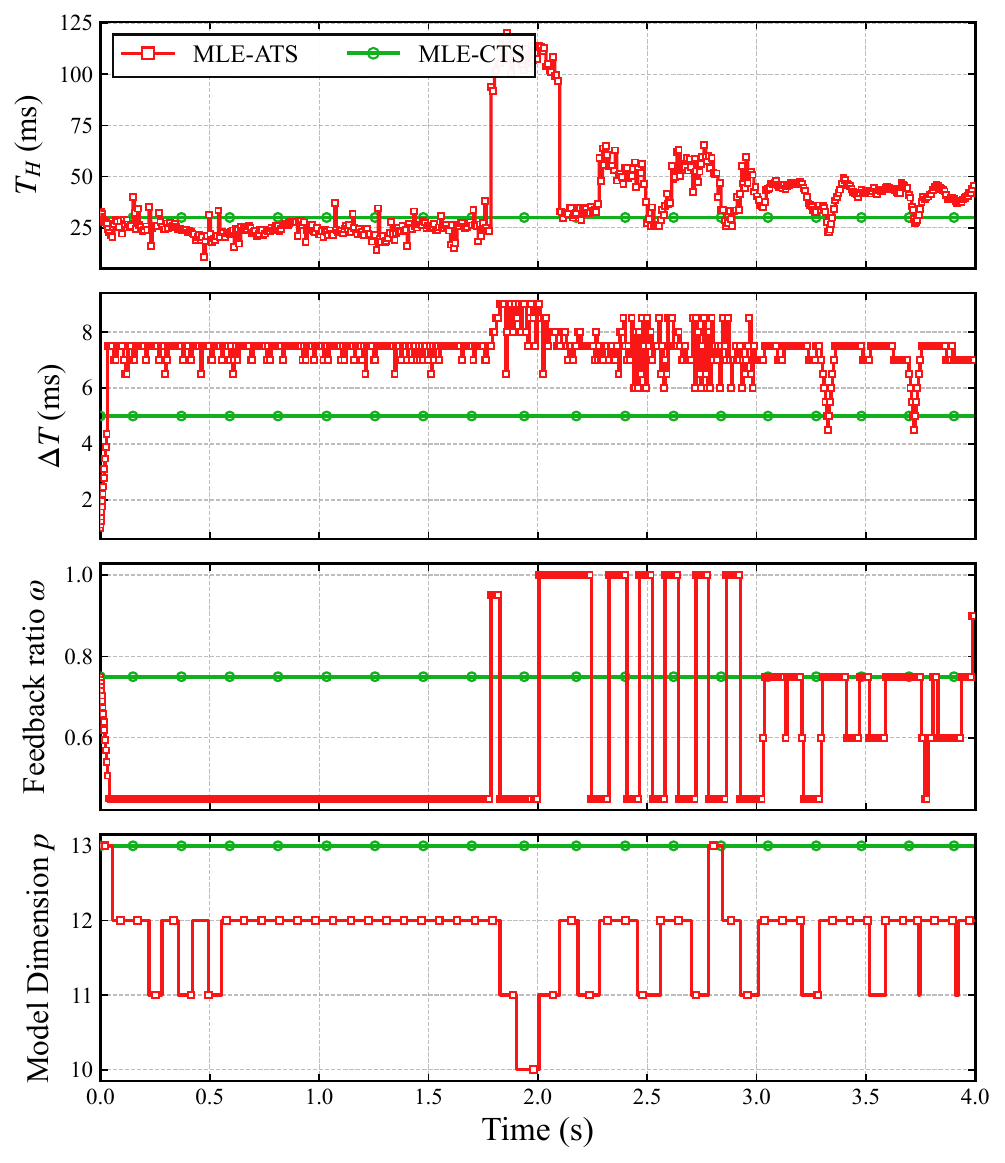}
    \caption{Tracking parameters of MLE-ATS and MLE-CTS at $20$ dB under the sharp-turn trajectory.}
    \label{fig:parameters_sharp}
\end{figure}


To evaluate robustness under highly nonstationary mobility, we use a trajectory with several sharp turns. The UE's mean and maximum speeds are $2.00$ m/s and $3.20$ m/s, respectively, and the environment contains one stationary scatterer. We set $R_{\mathrm{dim}}=R_{\mathrm{fb}}=5\times10^{-13}$, $\varepsilon_T=6.67$ ms (200 symbol durations), $\rho_\Delta=3.5\times10^{12}$, and $\rho_\omega=2.7\times10^{12}$ at an SNR of $20$ dB. \Cref{fig:gain_sharp,fig:traj_sharp} show that constant-parameter MLE-CTS follows the locally smooth segments but deviates around the sharp turns, where a fixed observation window cannot maintain an accurate local fit. MLE-ATS remains close to the ground-truth trajectory and the full-CSI benchmark by coordinating the observation-window length, model dimension, update interval, and feedback ratio. The coherence-time, MLVS, and EKF baselines yield delayed or scattered estimates under abrupt motion or model mismatch, while the no-tracking benchmark deteriorates once the initial beam becomes misaligned with the UE.

\Cref{fig:avg_gain_sharp} shows that MLE-ATS provides the highest mean gain among the practical schemes and remains close to the full-CSI benchmark, at approximately $0.97$ throughout the SNR sweep. Constant-parameter MLE-CTS increases from about $0.66$ to $0.77$. At $-5$ dB, MLE-ATS achieves approximately $0.97$, compared with $0.75$ for MLE-CTS and $0.60$ for the coherence-time method; MLVS and EKF remain close to zero. At high SNR, the coherence-time scheme benefits from more reliable local updates and approaches MLE-ATS, whereas MLVS remains lower because its fixed coherent processing interval responds more slowly to abrupt changes.

\Cref{fig:sharp_var_snr} shows that MLE-ATS has the smallest payload-accounted normalized-gain variance, on the order of $10^{-5}$, despite the nonstationary trajectory. The variance of MLE-CTS decreases from the $10^{-3}$ range toward $10^{-4}$ as SNR increases but remains above that of MLE-ATS. The MLVS, EKF, and coherence-time baselines exhibit substantially larger low-SNR variances because of local-sweeping interruptions or sensitivity to velocity-state and motion-model errors.

\Cref{fig:outage_sharp} further distinguishes the schemes in terms of per-symbol reliability. MLE-ATS maintains an effective-symbol ratio of $1.00$ across the full SNR range. At $0$ dB, the coherence-time and MLE-CTS schemes achieve $0.53$ and $0.43$, respectively; at $20$ dB, these ratios increase to $0.72$ and $0.69$. MLVS reaches $0.53$ at $25$ dB, whereas the EKF and no-tracking baselines remain below $0.15$. These results show that adaptive parameter control preserves the gain threshold through abrupt motion rather than only improving the averaged gain over time.

\Cref{fig:parameters_sharp} compares the adaptive and constant-parameter configurations under the sharp-turn trajectory at an SNR of $20$ dB. MLE-CTS maintains $T_H=33.33$ ms, $\Delta T=5.00$ ms, $\omega=0.75$, and $p=13$ throughout the tracking process. In contrast, MLE-ATS adjusts the observation-window length and polynomial model dimension according to the corrected estimation-risk scores, shortens the update interval when the tail estimation-risk score exceeds the head score, and controls the feedback ratio using the whole-interval estimation-risk score. Around sharp turns, the selected dimension may increase or decrease depending on the local data regime and the number of available feedback observations.


\section{Extension to UPA}
\label{sec:upa}
We extend the proposed framework from ULA-based two-dimensional
angle--range tracking to UPA-based three-dimensional
azimuth--elevation--range tracking. The MLE formulation, posterior
approximation, and adaptive parameter-adjustment structure extend by
replacing the trajectory parameter vector with its three-dimensional
counterpart.

\subsection{UPA Channel Model}
Consider an $N_y\times N_z$ UPA with $N=N_yN_z$ antenna elements.
The array is centered at the origin and lies on the $y$--$z$ plane.
The position of the $(i,j)$-th antenna element is
\(
\mathbf u_{i,j}
=
[0,\delta_i d,\delta_j d]^{\mathsf T},
\) where
\(
\delta_i=\frac{2i-N_y+1}{2},
\)
\(
\delta_j=\frac{2j-N_z+1}{2},
\) $i\in\{0,\ldots,N_y-1\}$ and
$j\in\{0,\ldots,N_z-1\}$.
In spherical coordinates, the propagation distance between a UE at range $r$, azimuth $\theta$, and elevation
$\phi$ and the
$(i,j)$-th antenna element is $r_{i,j}(r,\theta,\phi) = \|\mathbf p(r,\theta,\phi)-\mathbf u_{i,j}\|_2$
or equivalently,
{\small\( r_{i,j}(r,\theta,\phi) \triangleq \Big( r^2+(\delta_i^2+\delta_j^2)d^2 -2rd\big(\delta_i\cos\phi\sin\theta+\delta_j\sin\phi\big) \Big)^{\frac{1}{2}}, \)}
where $\theta$ denotes the azimuth angle, $\phi$ denotes the
elevation angle, and the Cartesian coordinate of the UE is $\mathbf{p}(r, \theta, \phi) = [r \cos \phi \cos \theta, r \cos \phi \sin \theta, r \sin \phi]^{\mathsf T}$.
Accordingly, the UPA near-field steering vector
$\mathbf b_{\mathrm{UPA}}(r,\theta,\phi)\in\mathbb C^N$
is defined element-wise as
\begin{equation}\label{eq:upa_b_def}
[
\mathbf b_{\mathrm{UPA}}(r,\theta,\phi)
]_{q(i,j)}
=
\frac{1}{\sqrt N}
\exp
\left(
-j\frac{2\pi}{\lambda}
(r_{i,j}(r,\theta,\phi)-r)
\right),
\end{equation}
where the row-major vectorization index is
$
q(i,j)=iN_z+j+1 .
$
The UE position is represented by the time-varying spherical
coordinates $(r(t),\theta(t),\phi(t))$. Following the ULA channel
model, the LoS component is
\begin{equation}\label{eq:upa_los_vec}
\mathbf h^{\mathrm{LoS}}_{t,\mathrm{UPA}}
=
g(t)
e^{-j\frac{2\pi r(t)}{\lambda}}
\mathbf b_{\mathrm{UPA}}
(r(t),\theta(t),\phi(t)),
\end{equation}
where
\(
g(t)=\frac{\lambda}{4\pi r(t)}
\)
is the free-space path gain.
For the NLoS component, assuming stationary scatterers, the reflected
paths are modeled as
\begin{equation}
\mathbf h^{\mathrm{NLoS}}_{t,\mathrm{UPA}}
=
\sum_{l=1}^{L-1}
g_l(t)
e^{-j\frac{2\pi}{\lambda}
(r_{l,1}+r_{l,2}(t))}
\mathbf b_{\mathrm{UPA}}
(r_{l,1},\theta_{l,1},\phi_{l,1}),
\end{equation}
where
\(
g_l(t)
=
\frac{\lambda p_l}
{4\pi (r_{l,1}+r_{l,2}(t))} .
\)
The complete UPA channel is therefore denoted by $\mathbf h_{t,\mathrm{UPA}}
=
\mathbf h^{\mathrm{LoS}}_{t,\mathrm{UPA}}
+
\mathbf h^{\mathrm{NLoS}}_{t,\mathrm{UPA}} .$
Similar to the ULA setting, the estimator focuses on the dominant
LoS component, while the weaker NLoS contribution is treated as model
mismatch.

\subsection{Extension of MLE-CTS and MLE-ATS}
The UPA extension introduces an additional elevation trajectory
dimension. Within the observation window used at update $m$, the local azimuth, range, and elevation trajectories are represented as $\widehat{\theta}_m(t')=\sum_{\ell=0}^{p_{\theta,m}}\alpha_{m,\ell}t'^{\ell},\quad \widehat r_m(t')=\sum_{\ell=0}^{p_{r,m}}\beta_{m,\ell}t'^{\ell},\quad \widehat{\phi}_m(t')=\sum_{\ell=0}^{p_{\phi,m}}\gamma_{m,\ell}t'^{\ell}.$  
We define the corresponding coefficient vectors as $\boldsymbol\alpha_m \triangleq [\alpha_{m,0},\ldots,\alpha_{m,p_{\theta,m}}]^{\mathsf T}$, 
$\boldsymbol\beta_m \triangleq [\beta_{m,0},\ldots,\beta_{m,p_{r,m}}]^{\mathsf T}$, and
$\boldsymbol\gamma_m \triangleq [\gamma_{m,0},\ldots,\gamma_{m,p_{\phi,m}}]^{\mathsf T}$.

Therefore, the UPA trajectory-parameter vector is
\( \boldsymbol\eta_m^{\mathrm{UPA}} \triangleq \left[ \boldsymbol\alpha_m^{\mathsf T}, \boldsymbol\beta_m^{\mathsf T}, \boldsymbol\gamma_m^{\mathsf T} \right]^{\mathsf T} \in \mathbb R^{p_m^{\mathrm{UPA}}}, \)
where
\(
p_m^{\mathrm{UPA}}
=
p_{\theta,m}
+
p_{r,m}
+
p_{\phi,m}
+
3.
\)
For an arbitrary UPA trajectory-parameter vector $\boldsymbol\eta^{\mathrm{UPA}}$, evaluating $\widehat{\theta}_m(t')$, $\widehat r_m(t')$, and $\widehat{\phi}_m(t')$ at $t'_{u,m}$ and substituting the result into \Cref{eq:upa_los_vec} gives the predicted channel $\widehat{\mathbf h}_{m,\mathrm{UPA}}(t'_{u,m};\boldsymbol\eta^{\mathrm{UPA}})$. The corresponding noise-free received value is
\( \mu_{u,\mathrm{UPA}} \left( \boldsymbol\eta^{\mathrm{UPA}} \right) \triangleq \widehat{\mathbf h}_{m,\mathrm{UPA}}^{\mathsf H} (t'_{u,m};\boldsymbol\eta^{\mathrm{UPA}}) \mathbf w_u x_u. \)
The UPA MLE at update $m$ is therefore
\( \widehat{\boldsymbol\eta}_m^{\mathrm{UPA}} = \underset{ \boldsymbol\eta^{\mathrm{UPA}} \in \Theta_m^{\mathrm{UPA}} }{ \arg\min} \; J_m^{\mathrm{UPA}} \left( \boldsymbol\eta^{\mathrm{UPA}} \right), \)
where $\Theta_m^{\mathrm{UPA}}$ denotes the feasible UPA parameter set associated with dimension $p_m^{\mathrm{UPA}}$, and
\( J_m^{\mathrm{UPA}} \left( \boldsymbol\eta^{\mathrm{UPA}} \right) \triangleq \sum_{u\in\mathcal U_m} \left| y_u - \mu_{u,\mathrm{UPA}} \left( \boldsymbol\eta^{\mathrm{UPA}} \right) \right|^2. \)
The optimization is performed using the same Adam-based procedure as
in the ULA setting, with the corresponding closed-form gradients
provided in Appendix~\ref{app:grad}.

The Gaussian posterior approximation becomes
\( p \left( \boldsymbol\eta_m^{\mathrm{UPA}} \mid \mathcal H_m \right) \approx \mathcal N \left( \widehat{\boldsymbol\eta}_m^{\mathrm{UPA}}, \boldsymbol\Sigma_{\boldsymbol\eta,m}^{\mathrm{UPA}} \right), \)
where
\( \boldsymbol\Sigma_{\boldsymbol\eta,m}^{\mathrm{UPA}} = \left[ \frac{1}{\sigma^2} \nabla_{\boldsymbol\eta^{\mathrm{UPA}}}^{2} J_m^{\mathrm{UPA}} \left( \widehat{\boldsymbol\eta}_m^{\mathrm{UPA}} \right) \right]^{-1}. \)
During interval $m+1$, a UPA trajectory hypothesis is drawn at each
designated feedback position $k$ as
\( \widetilde{\boldsymbol\eta}_{m,k}^{\mathrm{UPA}} \sim \mathcal N \left( \widehat{\boldsymbol\eta}_m^{\mathrm{UPA}}, \boldsymbol\Sigma_{\boldsymbol\eta,m}^{\mathrm{UPA}} \right). \)
For $u=u(m+1,k)$, the sampled parameter vector produces the channel hypothesis $ \widetilde{\mathbf h}_{u,\mathrm{UPA}}
\triangleq
\widehat{\mathbf h}_{m,\mathrm{UPA}}
(t'_{u,m};\widetilde{\boldsymbol\eta}_{m,k}^{\mathrm{UPA}}) $, 
and the corresponding TS payload beam is
$\mathbf w_u^{\mathrm{TS}} =
\widetilde{\mathbf h}_{u,\mathrm{UPA}}/
\|\widetilde{\mathbf h}_{u,\mathrm{UPA}}\|_2.$
At non-feedback positions, the beam is constructed from
$\widehat{\boldsymbol\eta}_m^{\mathrm{UPA}}$
instead.
 
The adaptive tracking-parameter adjustment in
\Cref{sec:adaptive} extends directly by replacing the ULA parameter vector and model
dimension with $\boldsymbol\eta_m^{\rm UPA}$ and
$p_m^{\rm UPA}$, respectively. We set $p_{\theta,m}=p_{\phi,m}=\lfloor(p_m^{\rm UPA}-3)/3\rfloor$ and
$p_{r,m}=p_m^{\rm UPA}-3-p_{\theta,m}-p_{\phi,m}$.
Because Lemmas~1 and 2 and the subset-optimism expansion are
stated for a generic real-valued trajectory-parameter vector, the corrected candidate score, subset-information correction, and update rules for $\Delta T_m$ and $\omega_m$ extend directly to the UPA model after replacing the ULA parameter vector and gradients with their UPA counterparts.

\subsection{UPA Simulation Results}

\begin{figure}[t]
    \centering
    \includegraphics[width=0.8\linewidth]{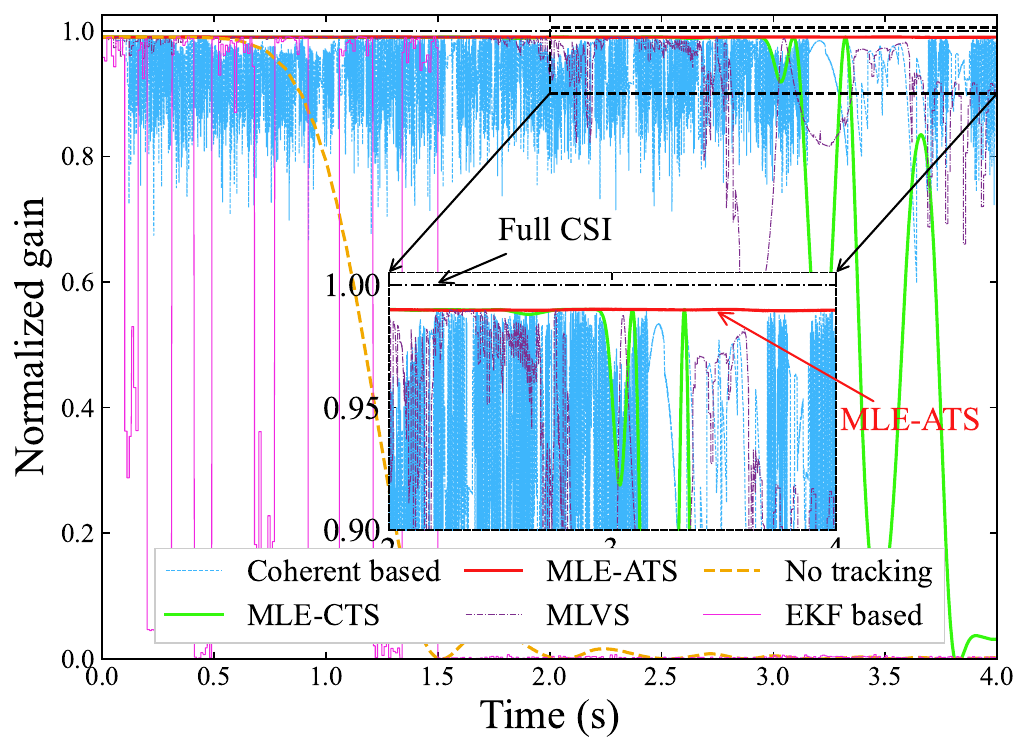}
    \caption{Normalized beamforming gain over time for three-dimensional UPA tracking.}
    \label{fig:gain_UPA}
\end{figure}

\begin{figure}[t]
    \centering
    \includegraphics[width=0.8\linewidth]{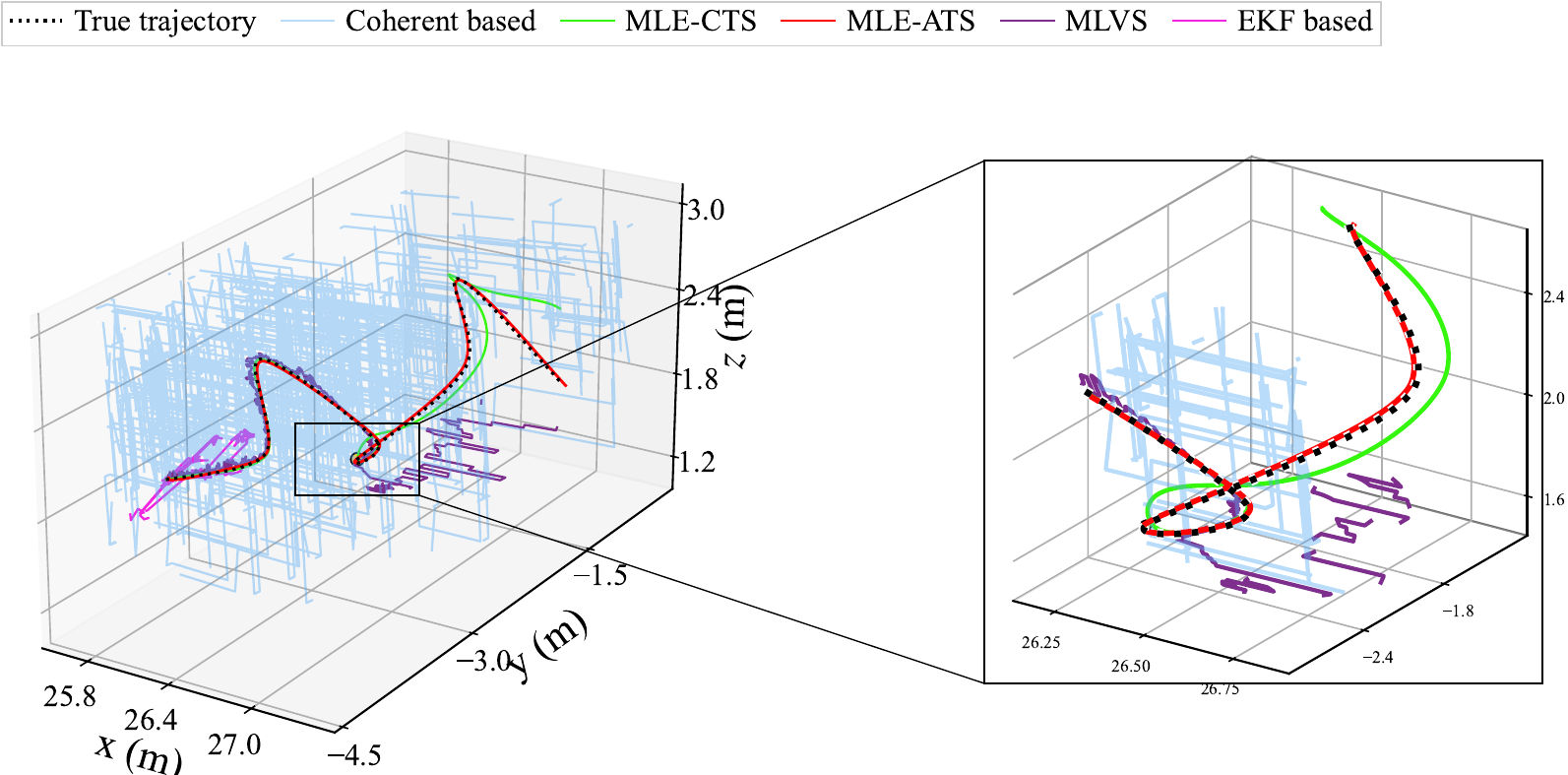}
    \caption{Three-dimensional sharp-turn trajectory-tracking comparison under the UPA setting.}
    \label{fig:traj_UPA}
\end{figure}

The simulation uses an $N_y\times N_z=128\times8$ UPA and a three-dimensional UE trajectory. The carrier frequency,
simulation duration, and symbol duration are identical to the ULA
setting.
As shown in \Cref{fig:gain_UPA,fig:traj_UPA}, MLE-ATS maintains a high normalized beamforming gain close to the full-CSI baseline. The MLE-CTS gain degrades when trajectory variations become pronounced, while the coherence-time-based scheme exhibits substantial trajectory error because of local beam sweeping. These results indicate that the adaptive tracking mechanism remains effective after incorporating elevation.

\section{Conclusion}\label{sec:con}

In this paper, we developed an adaptive payload-aided near-field beam-tracking framework based on MLE and TS. Local polynomial trajectory models are estimated from selected received payload samples, while a Gaussian approximation centered at the MLE and constructed from the observed Fisher information serves as a local surrogate for trajectory uncertainty. 
Posterior-sampled payload beams are used at
feedback positions for uncertainty-aware probing, and MLE-predicted
beams are used otherwise for exploitation. We further developed an
adaptive tracking-parameter mechanism that coordinates the
observation-window length, trajectory-model dimension, update
interval, and feedback ratio using estimation-risk scaling and
received-signal residual statistics. ULA simulations demonstrate high payload-accounted mean normalized beamforming gain, low payload-accounted normalized-gain variance, and high effective-symbol ratios under the considered trajectories, with substantial robustness gains under sharp-turn mobility. The UPA results demonstrate the extension to three-dimensional azimuth--elevation--range tracking. The proposed framework therefore
maintains near-field alignment from payload observations without
inserting dedicated non-payload beam-sweeping symbols during
steady-state tracking.


\appendices
\section{Closed-Form Gradient Derivation}
\label{app:grad}

For both ULA and UPA, let $e_u\triangleq y_u-\mu_u$ for $u\in\mathcal U_m$, so that $J_m=\sum_{u\in\mathcal U_m}|e_u|^2$.
Then, for any parameter block $\boldsymbol{\xi}$,
\begin{equation}
\nabla_{\boldsymbol{\xi}}J_m
=
-2\Re\!\left\{
\sum_{u\in\mathcal U_m}
e_u^*
\left(
\frac{\partial \hat{\mathbf h}_u}{\partial \boldsymbol{\xi}}
\right)^{\mathsf H}\mathbf w_ux_u
\right\}.
\label{eq:app_grad_generic}
\end{equation}

\subsection{ULA}
\label{app:grad_ula}

Let $s_{u,m}\triangleq t'_{u,m}$, $\widehat\theta_u\triangleq\widehat\theta_m(s_{u,m})$, $\widehat r_u\triangleq\widehat r_m(s_{u,m})$, and $a_n\triangleq\delta_nd$. Define
\begin{equation}
\hat\rho_{u,n}
\triangleq
\sqrt{\hat r_u^2+a_n^2-2\hat r_u a_n\sin\hat\theta_u}.
\end{equation}
Using \eqref{eq:h_est}, the $n$-th entry of the predicted ULA channel is $\hat{\mathbf h}_{u,\mathrm{ULA}}(n)
=
\frac{\lambda}{4\pi\sqrt N\,\hat r_u}
\exp\!\left(-j\kappa \hat\rho_{u,n}\right),
\kappa\triangleq\frac{2\pi}{\lambda}$.

Its derivatives with respect to $\hat\theta_u$ and $\hat r_u$ are
\begin{equation}
\frac{\partial \hat{\mathbf h}_{u,\mathrm{ULA}}(n)}{\partial \hat\theta_u}
=
j\kappa\,\hat{\mathbf h}_{u,\mathrm{ULA}}(n)
\frac{\hat r_u a_n\cos\hat\theta_u}{\hat\rho_{u,n}},
\end{equation}
\begin{equation}
\frac{\partial \hat{\mathbf h}_{u,\mathrm{ULA}}(n)}{\partial \hat r_u}
=
\hat{\mathbf h}_{u,\mathrm{ULA}}(n)
\left(
-\frac{1}{\hat r_u}
-j\kappa\frac{\hat r_u-a_n\sin\hat\theta_u}{\hat\rho_{u,n}}
\right).
\end{equation}

Since $\widehat\theta_m(t')=\sum_{\ell=0}^{p_{\theta,m}}\alpha_{m,\ell}t'^\ell$ and
$\widehat r_m(t')=\sum_{\ell=0}^{p_{r,m}}\beta_{m,\ell}t'^\ell$,
we have
\begin{equation}
\frac{\partial \hat{\mathbf h}_{u,\mathrm{ULA}}(n)}{\partial \alpha_{m,\ell}}
=
\frac{\partial \hat{\mathbf h}_{u,\mathrm{ULA}}(n)}{\partial \hat\theta_u}s_{u,m}^\ell,
\end{equation}
\begin{equation}
\frac{\partial \hat{\mathbf h}_{u,\mathrm{ULA}}(n)}{\partial \beta_{m,\ell}}
=
\frac{\partial \hat{\mathbf h}_{u,\mathrm{ULA}}(n)}{\partial \hat r_u}s_{u,m}^\ell.
\end{equation}
Stacking the element-wise derivatives over $n\in\mathcal N$ and substituting them into \eqref{eq:app_grad_generic} yields the gradients with respect to $\boldsymbol{\alpha}$ and $\boldsymbol{\beta}$.

\subsection{UPA}
\label{app:grad_upa}

Let $s_{u,m}\triangleq t'_{u,m}$, $\widehat\theta_u\triangleq\widehat\theta_m(s_{u,m})$, $\widehat r_u\triangleq\widehat r_m(s_{u,m})$, and $\widehat\phi_u\triangleq\widehat\phi_m(s_{u,m})$. Define $q_{u,i,j}
\triangleq
\delta_i\cos\hat\phi_u\sin\hat\theta_u+\delta_j\sin\hat\phi_u$, and $\hat\rho_{u,i,j}
\triangleq
\sqrt{\hat r_u^2+d^2(\delta_i^2+\delta_j^2)-2\hat r_ud\,q_{u,i,j}}.$
From \Cref{eq:upa_los_vec}, the $(i,j)$-th entry of the predicted UPA channel is $\hat{\mathbf h}_{u,\mathrm{UPA}}(i,j)
=
\frac{\lambda}{4\pi\sqrt N\,\hat r_u}
\exp\!\left(-j\kappa \hat\rho_{u,i,j}\right),
\kappa\triangleq\frac{2\pi}{\lambda}$.
Its derivatives with respect to $\hat\theta_u$, $\hat r_u$, and $\hat\phi_u$ are
\begin{equation}
\frac{\partial \hat{\mathbf h}_{u,\mathrm{UPA}}(i,j)}{\partial \hat\theta_u}
=
j\kappa\,\hat{\mathbf h}_{u,\mathrm{UPA}}(i,j)
\frac{\hat r_ud\,\delta_i\cos\hat\phi_u\cos\hat\theta_u}{\hat\rho_{u,i,j}},
\end{equation}
\begin{equation}
\frac{\partial \hat{\mathbf h}_{u,\mathrm{UPA}}(i,j)}{\partial \hat r_u}
=
\hat{\mathbf h}_{u,\mathrm{UPA}}(i,j)
\left(
-\frac{1}{\hat r_u}
-j\kappa\frac{\hat r_u-d\,q_{u,i,j}}{\hat\rho_{u,i,j}}
\right),
\end{equation}
\begin{equation}
\begin{aligned}
\frac{\partial \hat{\mathbf h}_{u,\mathrm{UPA}}(i,j)}{\partial \hat\phi_u}
&=
-j\kappa\,\hat{\mathbf h}_{u,\mathrm{UPA}}(i,j) \\
&\quad\times
\frac{\hat r_ud\big(\delta_i\sin\hat\phi_u\sin\hat\theta_u-\delta_j\cos\hat\phi_u\big)}
{\hat\rho_{u,i,j}} .
\end{aligned}
\end{equation}

Since $\widehat\theta_m(t')=\sum_{\ell=0}^{p_{\theta,m}}\alpha_{m,\ell}t'^\ell$,
$\widehat r_m(t')=\sum_{\ell=0}^{p_{r,m}}\beta_{m,\ell}t'^\ell$, and
$\widehat\phi_m(t')=\sum_{\ell=0}^{p_{\phi,m}}\gamma_{m,\ell}t'^\ell$,
we have
\begin{equation}
\frac{\partial \hat{\mathbf h}_{u,\mathrm{UPA}}(i,j)}{\partial \alpha_{m,\ell}}
=
\frac{\partial \hat{\mathbf h}_{u,\mathrm{UPA}}(i,j)}{\partial \hat\theta_u}s_{u,m}^\ell .
\end{equation}
\begin{equation}
\frac{\partial \hat{\mathbf h}_{u,\mathrm{UPA}}(i,j)}{\partial \beta_{m,\ell}}
=
\frac{\partial \hat{\mathbf h}_{u,\mathrm{UPA}}(i,j)}{\partial \hat r_u}s_{u,m}^\ell .
\end{equation}
\begin{equation}
\frac{\partial \hat{\mathbf h}_{u,\mathrm{UPA}}(i,j)}{\partial \gamma_{m,\ell}}
=
\frac{\partial \hat{\mathbf h}_{u,\mathrm{UPA}}(i,j)}{\partial \hat\phi_u}s_{u,m}^\ell .
\end{equation}
Stacking the element-wise derivatives according to \eqref{eq:upa_b_def} and substituting them into \eqref{eq:app_grad_generic} yields the gradients with respect to $\boldsymbol{\alpha}$, $\boldsymbol{\beta}$, and $\boldsymbol{\gamma}$.

\section{Gaussian Approximation for the Posterior}
\label{app:laplace_fim}
We derive the Gaussian approximation for the posterior distribution in~\Cref{eq:post_dist,eq:fim}.
Based on \Cref{eq:likelihood}, under a locally flat prior, the posterior distribution satisfies
\( p(\boldsymbol\eta\mid\mathcal H_m)\propto\exp\left(-J_m(\boldsymbol\eta)/\sigma^2\right). \)
We show that $p(\cdot\mid\mathcal H_m)$ is approximated by
$\mathcal N(\widehat{\boldsymbol\eta}_m,\boldsymbol\Sigma_{\boldsymbol\eta,m})$, where
\( \boldsymbol\Sigma_{\boldsymbol\eta,m}\triangleq\left(\sigma^{-2}\nabla_{\boldsymbol\eta}^2J_m(\widehat{\boldsymbol\eta}_m)\right)^{-1}. \)
We work under the following assumptions.
\begin{enumerate}
    \item The Hessian matrix at the MLE $\widehat{\boldsymbol\eta}_m$ satisfies $|\mathcal{U}_m|^{-1} \nabla_{\boldsymbol{\eta}}^2 J_m \big( \widehat{\boldsymbol{\eta}}_m \big) \succeq \lambda_0 \mathbf I_{p_m}.$

    \item For each $u \in \mathcal{U}_m$, we have
    {\small\begin{align*}
        |y_u - \mu_u (\boldsymbol{\eta})| &\leq L (1 + \Vert \boldsymbol{\eta} - \widehat{\boldsymbol{\eta}}_m \Vert_2), \; \Vert\nabla_{\boldsymbol{\eta}} \mu_u (\boldsymbol{\eta})\Vert_2 \leq L,\\
        \Vert \nabla_{\boldsymbol{\eta}}^2 \mu_u (\boldsymbol{\eta})\Vert_2 &\leq L,  \; \Vert \nabla_{\boldsymbol{\eta}}^3 \mu_u (\boldsymbol{\eta})\Vert_2 \leq L.
    \end{align*}}
    \item The posterior distribution concentrates around the MLE. In particular, we have
    \begin{align*}
        \Big\{ \int \Vert \boldsymbol{\eta} - \widehat{\boldsymbol{\eta}}_m\Vert_2^6 p ( \boldsymbol{\eta} \mid \mathcal{H}_m) d \boldsymbol{\eta} \Big\}^{1/6} \leq c_0 / \sqrt{|\mathcal{U}_m|}.
    \end{align*}
\end{enumerate}
The first assumption requires the Fisher information to be non-degenerate, which is standard in likelihood-based estimation and can be verified empirically. The second assumption requires $\mu_u$ to be sufficiently smooth, which can be verified from the expressions above. The third assumption requires the posterior to concentrate around the MLE; see Theorem 3.1 of \cite{mou2024diffusion} for a proof of such results.

Under the above assumptions, the KL divergence satisfies
\begin{align}
D_{\mathrm{KL}}\!\left(p(\cdot\mid\mathcal H_m)\,\middle\|\,\mathcal N(\widehat{\boldsymbol\eta}_m,\boldsymbol\Sigma_{\boldsymbol\eta,m})\right)
\leq \frac{a_0}{|\mathcal U_m|},
\label{eq:gaussian-approx-kl}
\end{align}
for a constant $a_0>0$ depending only on $(\lambda_0,L,c_0,\sigma)$.
Define the gradient mismatch
\begin{equation}
\mathbf d_m(\boldsymbol\eta)
\triangleq
\nabla J_m(\boldsymbol\eta)
-\sigma^2\boldsymbol\Sigma_{\boldsymbol\eta,m}^{-1}
(\boldsymbol\eta-\widehat{\boldsymbol\eta}_m).
\end{equation}
To study the approximation error, we use the method in Proposition 3.4 of \cite{mou2024diffusion}, which implies
{\small\begin{multline}
D_{\mathrm{KL}}\!\left(p(\cdot\mid\mathcal H_m)\,\middle\|\,\mathcal N(\widehat{\boldsymbol\eta}_m,\boldsymbol\Sigma_{\boldsymbol\eta,m})\right)\\
\leq \frac{\lambda_{\max}(\boldsymbol\Sigma_{\boldsymbol\eta,m})}{\sigma^4}
\int \|\mathbf d_m(\boldsymbol\eta)\|_2^2
p(\boldsymbol\eta\mid\mathcal H_m)d\boldsymbol\eta.
\label{eq:bvm-thm}
\end{multline}}
By the first assumption,
\( \lambda_{\max}(\boldsymbol\Sigma_{\boldsymbol\eta,m})\leq\sigma^2/(\lambda_0|\mathcal U_m|). \)

Since $\sigma^2\boldsymbol\Sigma_{\boldsymbol\eta,m}^{-1}=\nabla^2J_m(\widehat{\boldsymbol\eta}_m)$, the gradient mismatch satisfies
\begin{multline}
\|\mathbf d_m(\boldsymbol\eta)\|_2
\leq\int_0^1
\Big\|\nabla^2J_m\big((1-\gamma)\boldsymbol\eta+\gamma\widehat{\boldsymbol\eta}_m\big)\\
{}-\nabla^2J_m(\widehat{\boldsymbol\eta}_m)\Big\|_2
\left\|\boldsymbol\eta-\widehat{\boldsymbol\eta}_m\right\|_2d\gamma.
\end{multline}
Using the second assumption, for any $\boldsymbol\eta$,
\begin{align}
\left\|\nabla^3J_m(\boldsymbol\eta)\right\|_2
\leq4L^2\left(1+\left\|\boldsymbol\eta-\widehat{\boldsymbol\eta}_m\right\|_2\right)|\mathcal U_m|.
\end{align}
Consequently,
\begin{equation}
\|\mathbf d_m(\boldsymbol\eta)\|_2
\leq
4L^2|\mathcal U_m|
\left(
\left\|\boldsymbol\eta-\widehat{\boldsymbol\eta}_m\right\|_2^2
+
\left\|\boldsymbol\eta-\widehat{\boldsymbol\eta}_m\right\|_2^3
\right).
\end{equation}
Finally, applying the third assumption and Jensen's inequality gives
\begin{align}
&\int\|\mathbf d_m(\boldsymbol\eta)\|_2^2p(\boldsymbol\eta\mid\mathcal H_m)d\boldsymbol\eta\nonumber\\
&\leq32L^4|\mathcal U_m|^2\int\Big(\|\boldsymbol\eta-\widehat{\boldsymbol\eta}_m\|_2^4+\|\boldsymbol\eta-\widehat{\boldsymbol\eta}_m\|_2^6\Big)p(\boldsymbol\eta\mid\mathcal H_m)d\boldsymbol\eta\nonumber\\
&\leq32L^4\left(c_0^4+\frac{c_0^6}{|\mathcal U_m|}\right)
\leq64L^4\max(1,c_0)^6.
\end{align}
Substitution into \Cref{eq:bvm-thm} yields
\begin{align}
D_{\mathrm{KL}}\!\left(p(\cdot\mid\mathcal H_m)\,\middle\|\,\mathcal N(\widehat{\boldsymbol\eta}_m,\boldsymbol\Sigma_{\boldsymbol\eta,m})\right)
\leq\frac{64L^4\max(1,c_0)^6}{\sigma^2\lambda_0|\mathcal U_m|},
\end{align}
which proves \Cref{eq:gaussian-approx-kl}.

\section{Proof of the Chi-Square Characterization of the Estimation Risk}
\label{app:chi-square}

We prove Lemma~\ref{lem:chi_square_risk} for a generic real-valued local trajectory-parameter vector; the result applies to both the ULA and UPA models. Let $\mathcal S_\nu$ be a feedback-index set with $|\mathcal S_\nu|=\nu$, and let the fixed-dimensional true local parameter be $\boldsymbol\eta^\star\in\mathbb R^p$. Under the correctly specified local-model assumption, $y_u=\mu_u(\boldsymbol\eta^\star)+n_u$ for $u\in\mathcal S_\nu$, where $n_u\sim\mathcal{CN}(0,\sigma^2)$. Let $\ell_\nu(\boldsymbol\eta)\triangleq-\sigma^{-2}\sum_{u\in\mathcal S_\nu}|y_u-\mu_u(\boldsymbol\eta)|^2+\mathrm{const}$ denote the log-likelihood.

Define $\boldsymbol g_u\triangleq\nabla_{\boldsymbol\eta}\mu_u(\boldsymbol\eta^\star)\in\mathbb C^p$ and $\boldsymbol\Gamma_\nu\triangleq\nu^{-1}\sum_{u\in\mathcal S_\nu}\operatorname{Re}\{\boldsymbol g_u\boldsymbol g_u^{\mathsf H}\}$. We assume that $\boldsymbol\eta^\star$ is an interior point, the MLE $\widehat{\boldsymbol\eta}_\nu$ is consistent, $p$ remains fixed, the functions $\mu_u(\boldsymbol\eta)$ have uniformly bounded derivatives up to third order near $\boldsymbol\eta^\star$, $\boldsymbol\Gamma_\nu\overset{p}{\longrightarrow}\boldsymbol\Gamma\succ\mathbf 0$, and the normalized score satisfies the central limit condition.

For the complex Gaussian observation model, the score central limit theorem and the normalized-Hessian convergence give
\begin{equation}\frac{1}{\sqrt{\nu}}\nabla_{\boldsymbol\eta}\ell_\nu(\boldsymbol\eta^\star)\Rightarrow\mathcal N\!\left(\mathbf 0,\frac{2}{\sigma^2}\boldsymbol\Gamma\right),\quad-\frac{1}{\nu}\nabla_{\boldsymbol\eta}^2\ell_\nu(\boldsymbol\eta^\star)\overset{p}{\longrightarrow}\frac{2}{\sigma^2}\boldsymbol\Gamma.\end{equation}

Applying a first-order expansion of the MLE score equation $\nabla_{\boldsymbol\eta}\ell_\nu(\widehat{\boldsymbol\eta}_\nu)=\mathbf 0$ around $\boldsymbol\eta^\star$ and using Slutsky's theorem yield
\begin{equation}\sqrt{\nu}\left(\widehat{\boldsymbol\eta}_\nu-\boldsymbol\eta^\star\right)\Rightarrow\mathcal N\!\left(\mathbf 0,\frac{\sigma^2}{2}\boldsymbol\Gamma^{-1}\right).\end{equation}

Let $\boldsymbol\Delta_\nu\triangleq\widehat{\boldsymbol\eta}_\nu-\boldsymbol\eta^\star$. A first-order Taylor expansion of the predicted noise-free response gives
\begin{equation}\mu_u(\widehat{\boldsymbol\eta}_\nu)-\mu_u(\boldsymbol\eta^\star)=\boldsymbol g_u^{\mathsf T}\boldsymbol\Delta_\nu+\varepsilon_{u,\nu}^{\mathrm{Tay}},\qquad|\varepsilon_{u,\nu}^{\mathrm{Tay}}|=O_p\!\left(\|\boldsymbol\Delta_\nu\|_2^2\right).\end{equation}

The asymptotic normality above implies $\boldsymbol\Delta_\nu=O_p(\nu^{-1/2})$. Hence, uniformly bounded second derivatives give $\nu^{-1}\sum_{u\in\mathcal S_\nu}|\varepsilon_{u,\nu}^{\mathrm{Tay}}|^2=o_p(\nu^{-1})$, while the corresponding cross term is also $o_p(\nu^{-1})$ by the Cauchy--Schwarz inequality. Because $\boldsymbol\Delta_\nu$ is real-valued and $\boldsymbol\Gamma_\nu=\nu^{-1}\sum_{u\in\mathcal S_\nu}\operatorname{Re}\{\boldsymbol g_u\boldsymbol g_u^{\mathsf H}\}$, the estimation risk therefore satisfies
\begin{equation}R_{\mathcal S_\nu}(\widehat{\boldsymbol\eta}_\nu)=\boldsymbol\Delta_\nu^{\mathsf T}\boldsymbol\Gamma_\nu\boldsymbol\Delta_\nu+o_p\!\left(\nu^{-1}\right).\end{equation}

Define $\boldsymbol z_\nu\triangleq\sqrt{2\nu}\boldsymbol\Gamma_\nu^{1/2}\boldsymbol\Delta_\nu/\sigma$. The MLE limit and $\boldsymbol\Gamma_\nu\overset{p}{\longrightarrow}\boldsymbol\Gamma$ imply $\boldsymbol z_\nu\Rightarrow\mathcal N(\mathbf 0,\mathbf I_p)$. Consequently,
\begin{equation}\frac{2\nu R_{\mathcal S_\nu}(\widehat{\boldsymbol\eta}_\nu)}{\sigma^2}=\boldsymbol z_\nu^{\mathsf T}\boldsymbol z_\nu+o_p(1)\Rightarrow\chi_p^2,\end{equation}
which proves \Cref{eq:risk_chi_square}.

For the expected-estimation-risk result, assume that there exists $\delta>0$ such that $\sup_{\nu\geq1}\mathbb E[(2\nu R_{\mathcal S_\nu}(\widehat{\boldsymbol\eta}_\nu)/\sigma^2)^{1+\delta}]<\infty$. This uniform-integrability condition, together with the convergence above and $\mathbb E[\chi_p^2]=p$, gives $\mathbb E\!\left[R_{\mathcal S_\nu}(\widehat{\boldsymbol\eta}_\nu)\right]=\frac{\sigma^2p}{2\nu}+o\!\left(\nu^{-1}\right),$ which proves \Cref{eq:risk_mean}.

We next derive the more general same-sample optimism correction
used in Eqs.~\eqref{eq:candidate_residual_score} and
\eqref{eq:controller_corrected_score}. Let $\mathcal S$ denote the fitting
set, let $\mathcal A$ denote an arbitrary evaluation set, and define $\mathbf I(\mathcal B)
\triangleq
\frac{2}{\sigma^2}
\sum_{u\in\mathcal B}
\operatorname{Re}
\left\{
\mathbf g_u\mathbf g_u^{\mathsf H}
\right\}.$
A first-order expansion of the MLE score equation and fitted
residuals gives
$\mathbb{E}\left[\frac{1}{|\mathcal A|}\sum_{u\in\mathcal A}\left|y_u-\mu_u\left(\widehat{\boldsymbol\eta}_{\mathcal S}\right)\right|^2-\sigma^2\right]
=
\mathbb{E}\left[R_{\mathcal A}\left(\widehat{\boldsymbol\eta}_{\mathcal S}\right)\right]
-\frac{\sigma^2}{|\mathcal A|}\operatorname{tr}\left[\mathbf I(\mathcal A\cap\mathcal S)\mathbf I(\mathcal S)^{-1}\right]
+o\left(|\mathcal S|^{-1}\right).$
Consequently, $\frac{1}{|\mathcal A|}\sum_{u\in\mathcal A}\left|y_u-\mu_u\left(\widehat{\boldsymbol\eta}_{\mathcal S}\right)\right|^2-\sigma^2+\frac{\sigma^2}{|\mathcal A|}\operatorname{tr}\left[\mathbf I(\mathcal A\cap\mathcal S)\mathbf I(\mathcal S)^{-1}\right]$ provides a first-order estimator of $R_{\mathcal A}(\widehat{\boldsymbol\eta}_{\mathcal S})$. When $\mathcal A=\mathcal S$, the trace term simplifies to $p$, yielding the familiar correction $\sigma^2p/|\mathcal S|$.
This expansion is candidate-wise and treats each model configuration as fixed prior to selection. Hence, the correction accounts for the within-candidate fitting optimism.

\printbibliography

\vfill

\end{document}